\documentclass[12pt]{amsart}
\usepackage{amssymb,amscd}
\usepackage{verbatim}

\usepackage{amsmath,amssymb,graphicx,mathrsfs} % my new
\usepackage[colorlinks=true,allcolors = blue]{hyperref} % my new

\usepackage{cancel} % gemini asked to add

\let\frak\mathfrak

\def\>{\relax\ifmmode\mskip.666667\thinmuskip\relax\else\kern.111111em\fi}
\def\<{\relax\ifmmode\mskip-.333333\thinmuskip\relax\else\kern-.0555556em\fi}
\def\vsk#1>{\vskip#1\baselineskip}
\def\vv#1>{\vadjust{\vsk#1>}\ignorespaces}
\def\vvn#1>{\vadjust{\nobreak\vsk#1>\nobreak}\ignorespaces}
 \let\alb\allowbreak

 \let\tsize\textstyle 
\let\sssize\scriptscriptstyle 
 \let\vp\vphantom \let\hp\hphantom

\let\Maketitle\maketitle
\def\maketitle{\Maketitle\thispagestyle{empty}\let\maketitle\empty}

\newtheorem{thm}{Theorem}[section]
\newtheorem{cor}[thm]{Corollary}
\newtheorem{lem}[thm]{Lemma}

\theoremstyle{definition} % My June 14 2017
\newtheorem{exmp}{Example}[section]

\numberwithin{equation}{section}

\theoremstyle{definition}
\newtheorem*{rem}{Remark}

\let\mc\mathcal
\let\nc\newcommand

\let\eps\varepsilon

\let\ka\kappa
\let\la\lambda
\let\La\Lambda

\let\phi\varphi
\let\si\sigma

\let\thi\vartheta

\def\Zeta{\mathrm Z}

\let\ox\otimes

\let\geq\geqslant
\let\le\leqslant
\let\leq\leqslant

\let\on\operatorname
\let\bi\bibitem
\let\bs\boldsymbol

\def\C{{\mathbb C}}
\def\Z{{\mathbb Z}}

\def\F{{\mc F}}

\def\Oc{{\mc O}}

\def\lsym#1{#1\alb\dots\relax#1\alb}
\def\lc{\lsym,}

\def\End{\on{End}}

\def\0{{\sssize0}}

\def\beq{\begin{equation}}
\def\eeq{\end{equation}}
\def\be{\begin{equation*}}
\def\ee{\end{equation*}}

\nc{\bea}{\begin{eqnarray*}}
\nc{\eea}{\end{eqnarray*}}
\nc{\bean}{\begin{eqnarray}}
\nc{\eean}{\end{eqnarray}}

\let\ga\gamma
\let\Ga\Gamma

\nc{\Il}{{\mc I_{\bs\la}}}
\nc{\bla}{{\bs\la}}
\nc{\Fla}{\F_\bla}
\nc{\tfl}{{T^*\Fla}}
\nc{\GL}{{GL_n(\C)}}
\nc{\GLC}{{GL_n(\C)\times\C^*}}

\def\KZ/{{\slshape KZ\/}}
\def\qKZ/{{\slshape qKZ\/}}
\def\XXX/{{\slshape XXX\/}}

\nc{\A}{{\mc C}}

\begin{document}

\hrule width0pt
\vsk->

\title[Monodromy Eigenvectors for Difference Equations with Root-of-Unity Step]
{Monodromy Eigenvectors for Difference Equations with Root-of-Unity Step}

\author[Vitaly Tarasov and Alexander Varchenko]
{Vitaly Tarasov$\>^\circ$ and Alexander Varchenko$\>^\star$}

\maketitle

\begin{center}
{\it $^{\star}\<$Department of Mathematics, University
of North Carolina at Chapel Hill\\ Chapel Hill, NC 27599-3250, USA\/}
\vsk.5>
{\it $^\circ$Department of Mathematical Sciences,
Indiana University,
402 North Blackford St,
\\
Indianapolis, IN 46202-3216, USA\/}
\end{center}

{\let\thefootnote\relax
\footnotetext{\vsk-.8>\noindent
$^\circ\<${\sl E\>-mail}:\enspace vtarasov@iu.edu\>,
supported in part by Simons Foundation grants \rlap{430235, 852996}
\\
$^\star\<${\sl E\>-mail}:\enspace anv@email.unc.edu\>,
supported in part by Simons Foundation grant TSM-00012774}}

\vsk>
{\leftskip3pc \rightskip\leftskip \parindent0pt \Small
{\it Key words\/}:
Discrete flat qKZ connection, commuting operators, integral representation,
discrete integral
\vsk.6>
{\it MSC\<\>2020\/}: 81R50, 33C70, 33D80
%33C70, 55N91, 33D80 %14N35, 53D45, 14D05, 33C70
\par}

\begin{abstract}

A qKZ-type discrete flat connection with multiplicative step $p$ has monodromy operators when $p$ is a root of unity. The monodromy operators are commuting endomorphisms of the connection. We construct their eigensections and eigenvalues. The construction is based on the presentation of the discrete flat connection as a discrete Gauss--Manin connection.

\end{abstract}

{\small\tableofcontents\par}

\setcounter{footnote}{0}
\renewcommand{\thefootnote}{\arabic{footnote}}

\section{Introduction}
\label{sec 1}

A qKZ-type discrete flat connection with multiplicative step is a basic
structure in representation theory, mathematical physics, and enumerative geometry of Nakajima varieties. If the step is a root of unity, the connection acquires the monodromy operators. The monodromy operators are commuting endomorphisms of the connection. We construct their eigensections and eigenvalues.
The construction is based on the presentation of the discrete flat connection as a discrete Gauss--Manin connection.

\smallskip
The general setting is as follows.
Let $V$ be a complex vector space and $p\in\C^\times$. Let $z=(z_1,\dots,z_n)$.
Suppose we are given $\on{GL}(V)$-valued functions $K_m(z;p)$, $m=1,\dots,n$, satisfying
\beq
\label{flatness}
K_m(z_1,\dots,pz_l,\dots,z_n;p)\,K_l(z;p)
=
K_l(z_1,\dots,pz_m,\dots,z_n;p)\,K_m(z;p).
\eeq
Then these operators define a discrete flat connection with step $p$ and fiber $V$ on the space with
coordinates $z$.

A flat section $I(z;p)$ is a $V$-valued solution of the system of equations:
\beq
\label{qqKZ}
I(z_1,\dots,pz_m,\dots,z_n;p)=K_m(z;p)\,I(z;p),
\qquad m=1,\dots,n.
\eeq
A basic problem is to find integral representations for flat sections, or equivalently,
to realize the discrete connection $(K_m(z;p))$ as a suitable discrete Gauss--Manin connection.

\smallskip

When $p=1$, the operators $(K_m(z;1))$ commute, and one seeks their joint eigenvectors and
eigenvalues.

\smallskip

Equations \eqref{qqKZ} arise naturally in several guises: as trigonometric qKZ equations in
representation theory, and as quantum difference equations in the equivariant $K$-theory of
Nakajima varieties. In the case $p=1$, the commuting operators $(K_m(z;1))$ specialize to
transfer matrices of the XXZ integrable chain and to operators of quantum multiplication in
that same $K$-theory; see, for instance, \cite{S, FR, EFK, O, OS, AO}.

\smallskip

Integral representations are known both
for qKZ equations and for quantum difference equations of
Nakajima varieties. Such a representation consists of a scalar
(master) function $F(t;z;p)$
and a $V$-valued (weight) function $W(t;z)$, where $t=(t_1,\dots,t_r)$ is a collection of
auxiliary variables. The corresponding integrals
\[
I(z;p)=\int F(t;z;p)\,W(t;z)\,dt
\]
solve \eqref{qqKZ}.

The operators $(K_m(z;1))$ are diagonalized by the Bethe ansatz as follows. One introduces
\[
\Phi_l(t;z;p)\,=\,
%\lim_{p\to 1}
\frac{F(t_1,\dots,pt_l,\dots,t_r;z;p)}{F(t;z;p)},\qquad
\Psi_m(t;z;p)\,=\,
%\lim_{p\to 1}
\frac{F(t;z_1,\dots,pz_m,\dots,z_n;p)}{F(t;z;p)}.
\]
If $(t^0;z^0)$ solves the system of Bethe ansatz equations:
\[
\Phi_l(t;z;1)\,=\,1\,,\qquad l=1,\dots,r,
\]
then $W(t^0;z^0)$ is an eigenvector of $K_m(z^0;1)$ with eigenvalue $\Psi_m(t^0;z^0;1)$:
\[
K_m(z^0;1)\,W(t^0;z^0)= \Psi_m(t^0;z^0;1)\cdot W(t^0;z^0),
\qquad m=1,\dots,n.
\]

\medskip

We call the discrete connection $k$-resonant when $p=p_0$ is a primitive root of unity of
order $k$. In this case, the monodromy operators are defined by the formula:
\[
\tilde K_m(z;p_0)=
K_m(z_1,\dots,p_0^{k-1}z_m,\dots,z_n;p_0)\cdots
K_m(z_1,\dots,p_0z_m,\dots,z_n;p_0)\,K_m(z;p_0)\,.
\]
These operators commute:
\[
\tilde K_l(z;p_0)\,\tilde K_m(z;p_0)
=
\tilde K_m(z;p_0)\,\tilde K_l(z;p_0),
\qquad 1\le l,m\le n,
\]
and each $\tilde K_m(z;p_0)$ is an endomorphism of the discrete connection:
\[
\tilde K_m(z_1,\dots,p_0z_l,\dots,z_n;p_0)\,K_l(z;p_0)
=
K_l(z;p_0)\,\tilde K_m(z;p_0).
\]

\smallskip

The main result of this paper is a construction of common eigenvectors and eigenvalues
of the monodromy operators, using the integral representation for the discrete flat connection $(K_m(z;p))$. Namely, define
\bea
\tilde \Phi_l(t;z;p_0)
=
\prod_{i=0}^{k-1} \Phi_l(t_1,\dots,p_0^it_l,\dots,t_r;z;p_0),\quad
\tilde \Phi_l(t;z;p_0) =\prod_{i=0}^{k-1}\Psi_m(t;z_1,\dots,p_0^iz_m,\dots,z_n;p_0)
\eea
We show that if $(t^0;z^0)$ satisfies
\beq
\label{gbae}
\tilde \Phi_l(t;z;p_0)=1, \qquad l=1,\dots,r,
\eeq
then the vector
\bean
\label{mx00}
I(t^0;z^0)
&=&
\sum_{j_1,\dots,j_r=0}^{k-1}
\!\!
W(p_0^{j_1}t_1^0,\dots,p_0^{j_r}t_r^0;z^0)
\\
\notag
&
\times
&
\prod_{i=1}^r\prod_{d_i=0}^{j_i-1}
\Phi_i\bigl(p_0^{j_1}t_1^0,\dots,p_0^{j_{i-1}}t_{i-1}^0,p_0^{d_i}t_i^0,t_{i+1}^0,\dots,t_r^0;z^0;p_0\bigr),
\eean
if nonzero, is an eigenvector of $\tilde K_m(z^0;p_0)$ with eigenvalue $\tilde \Psi_m(t^0;z^0;p_0)$:
\[
\tilde K_m(z^0;p_0)\,I(t^0;z^0)=\tilde \Psi_m(t^0;z^0;p_0)\,I(t^0;z^0),
\qquad m=1,\dots,n,
\]
see Corollary \ref{cor m eigen}.

Formula \eqref{mx00} shows that the vector $I(t^0;z^0)$ is not simply the value
$W(t^0;z^0)$ of the weight function at a Bethe ansatz solution
$(t^0;z^0)$, as in the case $p=1$.
Rather, it is a sum of weight-function values at shifted points in the $t$-variables, with
scalar coefficients recording the displacement from the initial point $(t^0;z^0)$.

The scalar functions $\tilde \Phi_l(t;z;p_0)$ and $\tilde \Psi_m(t;z;p_0)$ satisfy a special symmetry:
\bean
\label{sss}
\tilde \Phi_l(t;z;p_0)
&=&\Phi_l(t_1^k,\dots,t_r^k;z_1^k,\dots,z_n^k;1),
\\
\notag
\tilde \Psi_m(t;z;p_0)
&=&
\Psi_m(t_1^k,\dots,t_r^k;z_1^k,\dots,z_n^k;1),
\eean
where also some parameters in  $\Phi_l$ and $\Phi_m$ in \eqref{sss} are 
shifted by the transformation $y\mapsto y^k$.
See Lemma \ref{lem 7.3} for the precise statement
 on the symmetry \eqref{sss} and the Motivating Example in Section \ref{n=2 p=-1}.

Thus the Bethe ansatz picture for $p=1$ and the monodromy picture for a primitive root of
unity are closely parallel. In the case $p=1$, solving
\[
\Phi_l(t;z;1)=1,
\qquad l=1,\dots,r,
\]
produces an eigenvector $W(t^0;z^0)$ of the operators $(K_m(z^0;1))$ with eigenvalues
$(\Psi_m(t^0;z^0;1))$. For a primitive root $p_0$, one instead solves
\[
\Phi_l(t_1^k,\dots,t_r^k;z_1^k,\dots,z_n^k;1)=1,
\qquad l=1,\dots,r,
\]
with appropriately shifted parameters,
and obtains an eigenvector $I(t^0;z^0)$ of the monodromy operators $(\tilde K_m(z^0;p_0))$
with eigenvalues
\[
(\Psi_m\bigl((t_1^0)^k,\dots,(t_r^0)^k;(z_1^0)^k,\dots,(z_n^0)^k;1\bigr)),
\]
with appropriately shifted parameters.

\subsection{Motivating Example}
\label{n=2 p=-1}\rm

Let \,$n=2$\,. Let
\,$V=(M_{\La_1}\ox M_{\La_2})[\<\>\La_1+\La_2-2\>]$ \, be the two-dimensional
weight subspace of the tensor product of two Verma modules over the quantum
group $U_q(\frak{sl}_2)$. The space $V$ has the basis
\be
f^{(1,0)}=f\<\>v_{\La_1}\ox v_{\La_2}\>,\qquad
f^{(0,1)}=v_{\La_1}\ox f\<\>v_{\La_2}\>.
\ee
In this basis, the trigonometric $R$-matrices on $V$ are
\be
R^{(1,2)}_{\La_1,\La_2}(z)\,=\,\frac1{q^{\La_1+\La_2}\!-z}\,
\begin{pmatrix} \,q^{\La_1}\!-z\<\>q^{\La_2} & q^{2\La_2}\!-1\\[4pt]
\,z\>(q^{2\La_1}\!-1) & q^{\La_2}\!-q^{\La_1} z \end{pmatrix}
\ee
and
\be
R^{(2,1)}_{\La_2,\La_1}(z)\,=\,\frac1{q^{\La_1+\La_2}\!-z}\,
\begin{pmatrix} \,q^{\La_1}\!-z\<\>q^{\La_2} & z\>(q^{2\La_2}\!-1)\\[4pt]
\,q^{2\La_1}\!-1 & q^{\La_2}\!-q^{\La_1} z \end{pmatrix}.
\ee
The discrete qKZ connection is given by the operators
\begin{align*}
K_1(z_1,z_2;p)\,&{}=\,\ka^{(\La_1-h^{(1)}\<)\</2}\>R^{(1,2)}_{\La_1,\La_2}(z_1/z_2)\,,
\\
K_2(z_1,z_2;p)\,&{}=\,R^{(2,1)}_{\La_2,\La_1}(p\<\>z_2/z_1)\>\ka^{(\La_2-h^{(2)}\<)\</2}\>,
\end{align*}
where
\be
\ka^{(\La_1-h^{(1)}\<)\</2}\,=\,\begin{pmatrix} \,\ka & 0\,\\[2pt] \,0 & 1\,\end{pmatrix},\qquad
\ka^{(\La_2-h^{(2)}\<)\</2}\,=\,\begin{pmatrix} \,1 & 0\,\\[2pt] \,0 & \ka\,\end{pmatrix}.
\ee
In this example $t$ is a single variable, and
\begin{align*}
\Phi(t;z_1,z_2;p)\,&{}=\,\ka\,
\frac{(1-q^{\La_1}t/z_1)\>(1-q^{\La_2}t/z_2)}{(q^{\La_1}\!-p\<\>t/z_1)\>(q^{\La_2}\!-p\<\>t/z_2)}\;,
\\[4pt]
\Psi_m(t;z_1,z_2;p)\,&{}=\;\frac{q^{\La_m}\!-t/z_m}{1-q^{\La_m}t/(pz_m)}\,,\qquad m\,=\,1,2\,.
\end{align*}
The $V$-valued weight function is
\be
W(t;z_1,z_2)\,=\,\frac t{z_1}\>(q^{\La_2}\!-t/z_2) f^{(1,0)} +\>
\frac t{z_2}\>(1-q^{\La_1}t/z_1) f^{(0,1)}\>.
\ee
Let \,$p_0=-1$\,, so that $p_0^2=1$. Then the monodromy functions of the qKZ connection are
\begin{align*}
\tilde K_1(z_1,z_2;-1)\,&{}=\,K_1(-z_1,z_2;-1)\,K_1(z_1,z_2;-1)\,,
\\
\tilde K_2(z_1,z_2;-1)\,&{}=\,K_2(z_1,-z_2;-1)\,K_2(z_1,z_2;-1)\,,
\end{align*}
and
\begin{align*}
& \tilde\Phi(t;z_1,z_2;-1)\,=\,
\ka^2\,\frac{(1-q^{2\La_1}\<\>t^2\</z_1^2)\>(1-q^{2\La_2}\<\>t^2\</z_2^2)}
{(q^{2\La_1}\!-t^2\</z_1^2)\>(q^{2\La_2}\!-t^2\</z_2^2)}\;,
\\[4pt]
& \tilde\Psi_m(t;z_1,z_2;-1)\,=\;\frac{q^{2\La_m}\!-t^2\</z_m^2}{1-q^{2\La_m}\<\>t^2\</z_m^2}\;,
\qquad m\,=\,1,2\,.
\end{align*}

\smallskip
\noindent
Corollary \ref{cor m eigen} in this case reads as follows.
Let \,$(t^0; z_1^0,z_2^0)$ \,be a solution to the Bethe ansatz equation
\be
\ka^2\,\frac{(1-q^{2\La_1}\<\>t^2\</z_1^2)\>(1-q^{2\La_2}\<\>t^2\</z_2^2)}
{(q^{2\La_1}\!-t^2\</z_1^2)\>(q^{2\La_2}\!-t^2\</z_2^2)}\;=\,1\,.
\ee
Then the vector
\be
I(t^0;z_1^{0},z_2^0)\,=\,W(t^0;z_1^0,z_2^0)\>+\>
W(-\<\>t^0;z_1^0,z_2^0)\,\Phi(t^0;z_1^0,z_2^0;-1),
\ee
if nonzero, is an eigenvector of the monodromy functions \,$\tilde K_m(z_1^0,z_2^0;-1)$\,,
\,$m=1,2$\,, with respective eigenvalues
\be
\frac{q^{2\Lambda_m}-(t^0\!/z^0_m)^2}
{1-q^{2\La_m}\<\>(t^0\!/z^0_m)^2}\;.
\ee

\bigskip

Our construction raises the usual Bethe-ansatz-type questions:
\begin{itemize}
\item
Is the constructed vector nonzero?
\item Do all eigenvectors of the monodromy operators arise in this way?
\item Can one compute their norms?
\item Are the resulting eigenvectors orthogonal?
\end{itemize}
\medskip

The first two questions are answered positively in the first nontrivial example in Appendix \ref{appB}.

\medskip

The exposition is as follows.
\begin{itemize}
\item \textbf{Sections \ref{sec 2} and \ref{sec 3}:} We introduce additive $b$-periodic discrete flat connections with step $a$ and construct flat eigensections for them.
\item \textbf{Sections \ref{sec 4} and \ref{sec 5}:} We apply this construction to multiplicative
discrete flat connections with integral representations.
\end{itemize}

To keep the notation minimal, in this paper we focus on applications of the general construction
to the trigonometric qKZ difference connection
associated with a tensor product of Verma modules over
$U_q(\mathfrak{sl}_2)$ with multiplicative root-of-unity step.

\smallskip

As an illustration, in Section \ref{sec 6} we consider the case when $q = p^c$ for some integer $c$ and other parameters of the trigonometric qKZ connection taking special values. In this case, the functions $\tilde\Phi_l(t;z)$, $\tilde\Psi_m(t;z)$ are identically equal to 1. This makes the system of equations \eqref{gbae} trivial: every $(t^0; z^0)$ satisfies the system. In this case our method produces Laurent polynomial flat sections of the discrete connection, and our construction has a flavor similar to that of the construction of polynomial solutions to the KZ and qKZ equations in finite characteristic presented in \cite{SV, MV1}.

\smallskip

The paper has two appendices.
In Appendix~\ref{appA}, we describe the semiclassical limit $q\to 1$ of the qKZ construction. We show that the monodromy operators of the qKZ discrete connection converge to the cyclotomic Gaudin Hamiltonians associated with a tensor product of Verma modules over $\frak{sl}_2$, and that the limit of our eigenvector construction yields new proofs of the commutativity and diagonalization theorems for the cyclotomic Gaudin Hamiltonians due to \cite{VY}.

In Appendix~\ref{appB}, the qKZ discrete connection is considered on the first nontrivial weight subspace of a tensor product of Verma modules over $U_q(\frak{sl}_2)$. It is shown that, for generic values of the parameters, our construction produces a basis of common eigenvectors of the monodromy operators.

\smallskip

This paper is related to \cite{TV3, TV4}.
In \cite{TV3},
a qKZ-type additive discrete flat connection in characteristic $p$ is considered.
The connection has $p$-curvature operators. They
are commuting endomorphisms of the connection. A Bethe-ansatz-type
construction of eigensections and eigenvalues of the $p$-curvature operators is presented.
An eigensection is a discrete
hypergeometric sum over the finite lattice
$\Z^r/p\Z^r$.

In \cite{TV4},
a KZ-type differential connection is considered.
A Bethe-ansatz-type
construction of eigensections and eigenvalues of the $p$-curvature operators of the connection
is presented.
An eigensection is a characteristic $p$ hypergeometric integral.

\smallskip

This paper is related to the work of P. Koroteev and A.
Smirnov \cite{KS}, where the eigenvalues of the monodromy operators are discussed, though
not the eigenvectors.
That paper is written in the style of a proposal, however: several
objects and statements are introduced without full precision, and the proofs are only sketched.

\smallskip

\noindent\textbf{Acknowledgments.}
The second author thanks IH\'ES for its hospitality during May--June 2026, when this paper was developed.

\section{Invariant discrete connection}
\label{sec 2}

\subsection{Definition}

Let $a,b$ be nonzero complex numbers, $V$ a complex vector space,
$x=(x_1,\dots, x_n)$ variables.
Let $L_l(x)$, $l=1,\dots,n$, be $\on{GL}(V)$-valued functions such that:
\bean
\label{L com}
L_m(x_l +a)L_l(x)
&=&
L_l(x_m +a)L_m(x),
\qquad
1\leq l,m\leq n,
\\
\label{L per}
L_l(x_m + b)
&=& L_l(x), \qquad 1\leq l,m\leq n.
\eean
We say that the {\it transition} functions $(L_l(x))$ define
a {\it $b$-periodic discrete flat connection with
step $a$.}\footnote{We use the following abbreviation. Instead of
$(x_1,\dots, x_{m-1}, x_m+a, x_{m+1}, \dots, x_n)$
we write $(x_m+a)$, and instead of
$(s_1, \dots, s_{l-1}, s_l+a,s_{l+1},\dots, s_r)$ for some $r$ and $l$
we write $(s_l+a)$, and so on. We indicate only changed variables in the list of all variables. }

\medskip

A $V$-valued function $I(x)$ is a flat section if
\bean
\label{fL}
I(x_l +a) = L_l(x) I(x), \qquad l=1,\dots, n.
\eean

\subsection{Monodromy of resonance}

Let $k$ be a positive integer.
A $b$-periodic discrete flat connection with step $a$ is
$k$-resonant,
if
\bean
\label{ka=b}
b=ka .
\eean
Then the monodromy functions are defined as:
\bean
\label{L def}
\tilde L_m(x)
&=&
L_m(x_m + (k-1)a) \dots L_m(x_m + a) L_m(x), \qquad
m=1,\dots, n.
\eean

\begin{lem}
\label{lem 1.1}

For $1\leq l,m\leq n$, we have
\bean
\label{tL per}
\tilde L_l(x_m+b)
&=&
\tilde L_l(x),
\\
\label{lm=ml}
\tilde L_l(x_m+a) L_m(x)
& =&
L_m(x)\tilde L_l(x),
\eean
and the monodromy functions commute:
\bean
\label{LL=LL}
\tilde L_l(x) \tilde L_m(x) = \tilde L_m(x) \tilde L_l(x),
\qquad 1\leq l,m\leq n.
\eean
\end{lem}

\begin{proof}
Formula \eqref{tL per} follows from \eqref{L per}.
Formulas \eqref{lm=ml} and \eqref{LL=LL} follow from the flatness conditions \eqref{L com}.
\end{proof}

\begin{cor}
Let $I(x)$ be a flat section. Then
$\tilde L_l(x) I(x)$, $l=1,\dots,n$, are flat sections.
\end{cor}

Our goal is to construct flat sections that are eigenvectors of the commuting monodromy functions $(\tilde L_l(x))$.

\section{Eigensections of $k$-resonant $b$-periodic connection}
\label{sec 3}

\subsection{Discrete flat connection of rank 1}

Let $r, n$ be positive integers, $s=(s_1,\dots, s_r)$, $x=(x_1,\dots,x_n)$.
Let $a$ and $b$ be nonzero complex numbers.
A $b$-periodic
discrete flat connection of rank 1 with step $a$ on the space with coordinates $(s;x)$
is a collection of nonzero scalar functions
$\phi_l(s;x)$, $l=1,\dots,r,$ and $\psi_m(s;x),$ $m=1,\dots,n$, such that
the functions are $b$-periodic with respect to variables $s$ and $x$, and
satisfy
\bea
\phi_l(s_i+a;x)\,\phi_i(s;x)
&=&
\phi_i(s_l+a;x) \,\phi_l(s;x), \qquad
\ \
1\leq l,i\leq r,
\\
\psi_m(s;x_j+a) \,\psi_j(s;x)
&=&
\psi_j(s;x_m+a) \,\psi_m(s;x), \qquad 1\leq m,j\leq n,
\\
\psi_m(s_i+a;x) \,\phi_i(s;x)
&=&
\phi_i(s; x_m+a) \,\psi_m(s;x), \qquad 1\leq i\leq r, \quad
1\leq m\leq n.
\eea

Let $k$ be a positive integer.
The discrete connection of rank 1 is $k$-resonant if $b=ka$.
In that case, the monodromy functions are:
\bea
\tilde \phi_l(s;x)
=
\prod_{j=0}^{k-1} \phi_l(s_l+ja; x), \qquad
\tilde \psi_m(s;x)
=
\prod_{j=0}^{k-1} \psi_m(s; x_m+ja), \qquad
\eea
$l=1,\dots, r,$ $m=1,\dots, n.$
\begin{lem}
\label{lem 2.1}
The monodromy functions are $a$-periodic in variables $s$ and $x$:
\bea
\tilde \phi_l(s_i+a;x)
&=&
\tilde \phi_l(s;x),
\qquad
\tilde \phi_l(s;x_j+a)
=
\tilde \phi_l(s;x),
\\
\tilde \psi_m(s_i+a;x)
&=&
\tilde \psi_m(s;x),
\qquad
\tilde \psi_m(s;x_j+a)
=
\tilde \psi_m(s;x).
\eea

\end{lem}

\begin{proof}

The lemma is the result of the application of formula \eqref{lm=ml} to the discrete flat connection
of rank 1.
\end{proof}

The group $\Z^{ n}$ acts on the space with coordinates $x$ by shifting the
coordinates of a point $x^0$ by integer multiples of $a$. Let
$\mc O_{x^0}$ denote the orbit of $x^0$.
The group $\Z^{r+n}$ acts on the space with coordinates $(s;x)$.
Let
$\mc O_{(s^0;x^0)}$ denote the orbit of $(s^0;x^0)$.

By Lemma \ref{lem 2.1}, each of the functions $(\tilde \phi_l(s;x),\tilde \psi_m(s;x))$ is constant on every orbit $\mc O_{(s^0;x^0)}$.

\subsection{Parallel transport}

The transition functions $(\phi_l(s;x), \psi_m(s;x))$ define parallel transport of the fiber $\C$ over points
of an orbit $\mc O_{(s^0; x^0)}$. Namely, the transport of fibers
$\C_{(s^1;x^1)} \leftarrow \C_{(s_l^1+a;x^1)}$ and
$\C_{(s^1;x^1)} \leftarrow \C_{(s^1;x^1_m+a)}$ are multiplications by
$\phi_l(s^1;x^1)$ and $\psi_m(s^1; x^1)$, respectively.
These elementary transports and the flatness conditions determine a transport
$\C_{(s^1;x^1)} \leftarrow \C_{(s^2;x^2)}$ for any two points
$(s^1; x^1)$ and $(s^2; x^2)$ of $\mc O_{(s^0; x^0)}$
as multiplication by a uniquely determined number which we denote by
$\mu_{(s^1; x^1),(s^2; x^2)}$.

For example, for nonnegative integers $j_1,\dots,j_r$, we have:
\bean
\label{pt}
&&
\mu_{(s^0; x^0),(s^0_1+j_1a,\dots, s^0_r+j_ra; x^0)}
\\
\notag
&&
\phantom{aaa}
=\prod_{i=1}^r\prod_{d_i=0}^{j_i-1}\phi_i(s^0_1+j_1a,\dots, s^0_{i-1}+ j_{i-1}a,
s^0_i + d_ia, s^0_{i+1}, \dots, s^0_r; x^0).
\eean

For a $V$-valued function \>$f(s;x)$ and a point \,$x^1 \in \mc O_{x^0}$\,,
define the discrete integral over variables \>$s$ \>by the formula:
\be
\int_{(s^0; x^0)} f(s;x^1)\, d_a s\;:=\!
\sum_{j_1, \dots, j_r=0}^{k-1}
f(s_1^0+j_1a, \dots, s_r^0+j_ra; x^1)\ \mu_{(s^0; x^0),(s_1^0 + j_1a,
\dots, s_r^0 + j_ra; x^1)}\,.
\ee
In addition, define the discrete integral over variables \>$s$
\>with the fixed variable \>$s_l$\,,
\be
\int^{\>(l)}_{(s^0; x^0)} f(s;x^1)\, d_a s\;:=\!\!\!\!
\sum_{j_1,\dots,\widehat{j_{l}},\dots, j_r=0}^{k-1}\!\!\!f(s_1^0+j_1a, \dots, s_r^0+j_ra; x^1)
\;\mu_{(s^0; x^0),(s_1^0 + j_1a,\dots, s_r^0 + j_ra; x^1)}\,,
\ee
where in the right-hand side there is no summation over \,$j_l$\,, and \,$j_l=0$ \,instead.

\begin{lem}
\label{lem int d}
Let $g(s;x)$ be a $b$-periodic function in variables \,$s$ and \,$x$. Then
\beq
\label{ili d}
\int_{(s^0;x^0)} \left(\phi_l(s;x^1)g(s_l+a;x^1) - g(s;x^1)\right) d_as\,=\,
\bigl(\tilde\phi_l(s^0;x^0)-1\bigr)
\int^{(l)}_{(s^0; x^0)} g(s;x^1)\, d_a s
\eeq
for \,$x^1\in \mc O_{x^0}$ and \;$l=1,\dots,r$.
\end{lem}

The lemma is a discrete version of Stokes' theorem.

\begin{proof}
We prove \eqref{ili d} for $l=1$. The proof for other \>$l$ is similar.
We have
\bean
\label{d1}
&&
\int _{(s^0;x^0)}\phi_1(s;x^1)g(s_1+a;x^1) d_as\,=\!
\sum_{j_1, \dots, j_r=0}^{k-1}\phi_1(s_1^0+j_1a, \dots, s_r^0+j_ra; x^1)
\\
\notag
&&
\times \ \
g(s_1^0+(j_1+1)a, s_2^0+j_2a, \dots, s_r^0+j_ra; x^1)
\ \mu_{(s^0; x^0),(s_1^0 + j_1a,
\dots,
s_r^0 + j_ra; x^1)}\,
\\
\notag
&&
=
\sum_{j_1, \dots, j_r=0}^{k-1}
g(s_1^0+(j_1+1)a, s_2^0+j_2a, \dots, s_r^0+j_ra; x^1)
\\
\notag
&&
\phantom{aaaaaa}
\ \
\times \ \
\mu_{(s^0; x^0),(s_1^0 + (j_1+1)a, s_2^0 + j_2a,
\dots,
s_r^0 + j_ra; x^1)}\,.
\eean
The last sum equals
\begin{align*}
& \kern-2em
\sum_{j_1, \dots, j_r=0}^{k-1}
g(s_1^0+j_1a, s_2^0+j_2a, \dots, s_r^0+j_ra; x^1)
\;\mu_{(s^0; x^0),(s_1^0 + j_1a, s_2^0 + j_2a,
\dots, s_r^0 + j_ra; x^1)}
\\[-4pt]
& {}+\,\bigl(\tilde\phi_1(s^0;x^0)-1\bigr)\!\sum_{j_2, \dots, j_r=0}^{k-1}
g(s_1^0,s_2^0+j_2a, \dots, s_r^0+j_ra; x^1)\;
\mu_{(s^0; x^0),(s_1^0,s_2^0 + j_2a,\dots, s_r^0 + j_ra; x^1)}
\\
& {}=\,\int_{(s^0;x^0)} g(s;x^1)\, d_as\,+\,
\bigl(\tilde\phi_1(s^0;x^0)-1\bigr)\int^{(1)}_{(s^0; x^0)} g(s;x^1)\, d_a s\,,
\end{align*}
because for any \>$j_2, \dots, j_r$\,,
\be
g(s_1^0+ka, s_2^0+j_2a, \dots, s_r^0+j_ra; x^1)\,=\,
g(s_1^0, s_2^0+j_2a, \dots, s_r^0+j_ra; x^1)
\ee
due to \,$b\>$-periodicity of \,$g(s;x)$\,, \,and
\be
\mu_{(s^0; x^0),(s_1^0 + ka, s_2^0 + j_2a,\dots, s_r^0 + j_ra; x^1)}
\,=\,\tilde \phi_1(s^0; x^0)
\,\mu_{(s^0; x^0),(s_1^0, s_2^0 + j_2a,\dots, s_r^0 + j_ra; x^1)}
\ee
since \,$\tilde\phi_1(s;x)$ \,is constant on the orbit \,$\mc O_{(s^0;x^0)}$\,.
The lemma is proved.
\end{proof}

\subsection{Integral representation}
\label{sec 2.4}

Assume that $\on{GL}(V)$-valued functions $(L_l(x))$ define a $k$-resonant
$b$-periodic discrete flat connection with step $a$. Assume
that nonzero scalar functions $(\phi_l(s;x), \psi_m(s;x))$ define
a $k$-resonant $b$-periodic
discrete flat connection of rank 1 with step $a$.

Let $w(s;x)$ and $g_{l,m}(s;x)$, $l=1, \dots, r$, $m=1,\dots,n$, be $V$-valued functions
$b$-periodic with respect to variables $s$ and $x$. Assume that these functions satisfy the relations:
\beq
\label{IR}
\psi_m(s;x) w(s;x_m+a) = L_m(x)w(s;x) +
\sum_{l=1}^r \left(\phi_l(s;x)\,g_{l,m}(s_l+a;x) -g_{l,m}(s;x)\right),
\eeq
$m=1,\dots,n$. Such a collection of functions is called an {\it integral
representation} for the discrete flat connection $(L_l(x))$. The function
$w(s;x)$ is called the weight function.

\smallskip
Define a $V$-valued function \,$I : \mc O_{x^0}\to V$, \;$x^1 \mapsto I(x^1)$\,,
by the formula
\begin{align}
\label{ix1}
I(x^1)\,&{}=\, \int_{(s^0;x^0)} w(s;x^1)\, d_a s
\\[4pt]
\notag
&{}=\!\sum_{j_1, \dots, j_r=0}^{k-1}
w(s_1^0+j_1a, \dots, s_r^0+j_ra; x^1)\;\mu_{(s^0; x^0),(s_1^0 + j_1a,
\dots, s_r^0 + j_ra; x^1)}\,.
\end{align}

\begin{thm}
\label{thm sol0}
For \,$x^1\in \mc O_{x^0}$ \,and \,$m=1,\dots, n$\,,
\beq
\label{ili0}
I(x^1_m+a)\,=\,L_m(x^1)\,I(x^1)\,+\,\sum_{l=1}^r\,
\bigl(\tilde\phi_l(s^0;x^0)-1\bigr)\int^{(l)}_{(s^0; x^0)} g_{l,m}(s;x^1)\, d_a s\,.
\eeq
\end{thm}
\begin{proof}
We have relations \eqref{IR} by assumption. To prove relations \eqref{ili0}\,,
we send each term $T(s;x)$ of relations \eqref{IR} to the element
$\int_{(s^0;x^0)}T(s;x^1)\,d_a s\,\in V$. In particular, the left-hand side is sent to
\begin{align*}
& \int_{(s^0;x^0)} \psi_m(s;x^1) \,w(s;x_m^1+a)\, d_a s
\\[3pt]
&{}=\!\sum_{j_1, \dots, j_r=0}^{k-1}w(s_1^0+j_1a, \dots, s_r^0+j_ra; x^1_m+a)
\\
&\hp{{}=\!\sum_{j_1, \dots, j_r=0}}
\times\,\psi_m(s_1^0+j_1a, \dots, s_r^0+j_ra; x^1)\,
\mu_{(s^0; x^0),(s_1^0 + j_1a,\dots,s_r^0 + j_ra; x^1)}
\\[2pt]
&{}=\!\sum_{j_1, \dots, j_r=0}^{k-1} w(s_1^0+j_1a, \dots, s_r^0+j_ra; x^1_m+a)\,
\mu_{(s^0; x^0),(s_1^0 + j_1a,\dots, s_r^0 + j_ra; x_m^1+a)}\,=\,I(x_m^1+a)\,.
\end{align*}
The first term on the right-hand side is sent to
\be
\int_{(s^0;x^0)} L_m(x^1)w(s;x^1)\,d_as\,=\,
L_m(x^1)\int_{(s^0;x^0)} w(s;x^1)\,d_a s\,=\,L_m(x^1)\,I(x^1)\,,
\ee
and the last sum in \eqref{IR} is sent to
\be
\sum_{l=1}^r\,\bigl(\tilde\phi_l(s^0;x^0)-1\bigr)
\int^{(l)}_{(s^0; x^0)} g_{l,m}(s;x^1)\, d_a s
\ee
by Lemma \ref{lem int d}. Theorem \ref{thm sol0} is proved.
\end{proof}

\begin{thm}
\label{thm i sol0}
For \,$x^1\in \mc O_{x^0}$ \,and \;$m=1,\dots, n$\,,
\be
\tilde L_m(x^1)\,I(x^1)\,=\tilde\psi_m(s^0;x^0)\,I(x^1)\,-\>
\sum_{l=1}^r\,\bigl(\tilde\phi_l(s^0;x^0)-1\bigr)\,\tilde g_{l,m}(s^0;x^1)\,,
\ee
where
\beq
\label{glmt}
\tilde g_{l,m}(s;x)\,=\,\sum_{j=0}^{k-1}\,
L_m\bigl(x_m+(k-1)\>a\bigr)\dots\>L_m\bigl(x_m+(j+1)\>a\bigr)\>
\int^{(l)}_{(s^0; x^0)} g_{l,m}(s;x_m+j\<\>a)\,d_as\,.
\eeq
\end{thm}

\begin{proof}
By Theorem \ref{thm sol0},
\be
\tilde L_m(x^1)\,I(x^1)\,=\,I(x_m^1+k\<\>a)\,-\>
\sum_{l=1}^r\,\bigl(\tilde\phi_l(s^0;x^0)-1\bigr)\,\tilde g_{l,m}(s^0;x^1)
\ee
since \,$\tilde \phi_l(s;x)$ \,is constant on \,$\mc O_{(s^0;x^0)}$.
On the other hand,
\begin{align*}
I(x^1_m+ ka)\,&{}=\!
\sum_{j_1, \dots, j_r=0}^{k-1} w(s_1^0+j_1a, \dots, s_r^0+j_ra; x^1_m+ka)
\,\mu_{(s^0; x^0),(s_1^0 + j_1a,\dots,s_r^0 + j_ra; x^1_m+ka)}
\\[4pt]
&{}=\!\sum_{j_1, \dots, j_r=0}^{k-1} w(s_1^0+j_1a, \dots, s_r^0+j_ra; x^1_m+k\<\>a)
\,\mu_{(s^0; x^0),(s_1^0 + j_1a,\dots,s_r^0 + j_ra; x^1_m+ka)}\,.
\end{align*}
We have
\be
w(s_1^0+j_1a, \dots, s_r^0+j_ra; x^1_m+ka)\,=\,
w(s_1^0+j_1a, \dots, s_r^0+j_ra; x^1)
\ee
due to \,$b$-periodicity of \,$w(s;x)$\,, and
\begin{align*}
\mu_{(s^0; x^0),(s_1^0 + j_1a, \dots, s_r^0 + j_ra; x^1_m+ka)}\,&{}=\,
\mu_{(s^0; x^0),(s_1^0 + j_1a, \dots, s_r^0 + j_ra; x^1)}\,
\tilde \psi_m (s_1^0 + j_1a, \dots, s_r^0 + j_ra; x^1)
\\[3pt]
&{}=\,\mu_{(s^0; x^0),(s_1^0 + j_1a, \dots, s_r^0 + j_ra; x^1)}
\tilde \psi_m (s^0; x^0)
\end{align*}
since \,$\tilde \psi_m(s;x)$ \,is constant on \,$\mc O_{(s^0;x^0)}$. Thus
\be
I(x^1_m+ ka)\,=\,\tilde \psi_m (s^0; x^0)\,I(x^1)\,.
\ee
The theorem is proved.
\end{proof}

\subsection{Bethe ansatz equations}

The Bethe ansatz equations are the system of equations
\beq
\label{bae}
\tilde \phi_l(s;x) = 1, \qquad l=1,\dots,r.
\eeq
A solution $(s^0;x^0)$ is a point where the values of
the monodromy functions in the $s$-directions equal \,$1$\,.
If $(s^0; x^0)$ is a solution to the Bethe ansatz equations, then all points of
the orbit $\mc O_{(s^0;x^0)}$ are also solutions to the Bethe ansatz equations
by Lemma \ref{lem 2.1}.

\subsection{Flat sections}

Let $(s^0;x^0)$ be a solution to the Bethe ansatz equations \eqref{bae}\,.

\begin{thm}
\label{thm sol}
The $V$-valued function \,$I$ defined by \>\eqref{ix1} is a flat section over
\;$\mc O_{x^0}$ of the discrete flat connection \;$(L_m(x))$. Namely,
\beq
\label{ili}
I(x^1_m+a)\,=\,L_m(x^1)\,I(x^1)\,,\qquad m=1,\dots,n\,.
\eeq
\end{thm}
\begin{proof}
The statement follows from Theorem \ref{thm sol0} since the sum in
the right-hand side of equality \eqref{ili0} vanishes due to Bethe ansatz equations
\eqref{bae}\,.
\end{proof}

\subsection{Eigensections and eigenvalues}

\begin{thm}
\label{thm i sol}
The flat section \,$I : \mc O_{x^0}\to V$ of Theorem \ref{thm sol}, if nonzero,
is an eigensection of the mono\-dromy functions,
\beq
\label{e so}
\tilde L_m(x^1)\,I(x^1)\,=\,\tilde\psi_m(s^0;x^0)\,I(x^1)
\eeq
for \,$x^1\in\mc O_{x^0}$\,, \,$m=1,\dots,n$\,.
\end{thm}

Recall that \,$ \tilde \psi_m(s^0;x^0)$ \,is the value at $(s^0;x^0)$ of
the scalar monodromy function \,$\tilde \psi_m(s;x)$ \,of the discrete flat
connection of rank \,$1$\,, and the function \,$\tilde \psi_m(s;x)$ \,is constant
on the orbit \,$\mc O_{(s^0;x^0)}$\,.
\begin{proof}
The statement follows from Theorem \ref{thm i sol0} and Bethe ansatz equations
\eqref{bae}\,.
\end{proof}

\begin{cor}
\label{cor eigen}
Let $(s^0;x^0)$ be a solution to the Bethe ansatz equations.
Then the vector
\begin{align}
\notag
\\[-14pt]
\label{ix0}
I(x^0)\,&{}=\,\int_{(s^0;x^0)} w(s;x^0) d_a s
\\[4pt]
\notag
&{}=\!\sum_{j_1, \dots, j_r=0}^{k-1}
w(s_1^0+j_1a, \dots, s_r^0+j_ra; x^0)\;
\mu_{(s^0; x^0),(s_1^0 + j_1a,\dots,s_r^0 + j_ra; x^0)}\,,
\end{align}
if nonzero, is an eigenvector of the monodromy operators
$(\tilde L_m(x^0))$ with respective eigenvalues $(\tilde \psi_m(s^0;x^0))$.
The numbers $\mu_{(s^0;x^0),(s_1^0 + j_1a, \dots,s_r^0 + j_ra; x^0)}$
are given in \eqref{pt}.
\end{cor}

\subsection{Independence of $(s^0; x^0)$}

We began with an orbit $\mc O_{(s^0;x^0)}$ consisting of
solutions to the Bethe ansatz
equations.
Using an integral representation with a weight function $w(s;x)$, we defined
a function
$ I_{(s^0; x^0)}: \mc O_{x^0} \to V$,
\bea
I_{(s^0;x^0)}(x^1):= I(x^1) =\int_{(s^0;x^0)} w(s;x^1) d_as\,,
\eea
via formula \eqref{ix1}. We demonstrated that the function
$ I_{(s^0;x^0)}$, if nonzero, is a flat eigensection
of the monodromy operators of the discrete connection $(L_m(x))$.

\medskip

In a similar way, for any $(\bar s; \bar x)\in \mc O_{(s^0;x^0)}$,
we can define a function
$ I_{(\bar s; \bar x)}: \mc O_{x^0} \to V$,
\bea
I_{(\bar s; \bar x)}(x^1)=\int_{(\bar s; \bar x)} w(s;x^1) d_as\,.
\eea
By Theorems \ref{thm sol} and \ref{thm i sol}, the function
$I_{(\bar s;\bar x)}$, if nonzero, is also a flat eigensection of the monodromy operators with the same eigenvalues as
$ I_{(s^0;x^0)}$.

\begin{thm}
\label{thm ind}
The functions \,$I_{(s^0;x^0)}$ and \,$I_{(\bar s;\bar x)}$
are proportional with the nonzero coefficient of proportionality $\mu_{(s^0;x^0),(\bar s;
\bar x)}$. More precisely,
\beq
\label{ind s}
I_{(s^0;x^0)} (x^1)\,=\,\mu_{(s^0;x^0),(\bar s; \bar x)}\,I_{(\bar s;\bar x)}(x^1)\,,
\qquad x^1\in \mc O_{x^0}\,.
\eeq
\end{thm}

\begin{proof}
It is enough to prove the theorem when $(s^0;x^0)$ and $(\bar s\>; \bar x)$
differ by $a$ at one of the coordinates. If $\,\bar s = s^0$ and
$\,\bar x = (x^0_m+a)$, then clearly
\be
I_{(s^0;x^0)} (x^1)\,=\,\psi_m(s^0; x^0)\,I_{(s^0;x^0_m+a)}(x^1)\,,
\qquad x^1\in\mc O_{x^0}\,.
\ee
If $\bar s = (s^0_l+a)$ and $\bar x = x^0$, we need to prove that
\beq
\label{any m}
I_{(s^0;x^0)} (x^1) = \phi_l(s^0; x^0) I_{(s^0_l+a;x^0)}(x^1).
\eeq
We prove this statement for $l=1$. The proof for other $l$ is similar. We have:
\begin{alignat*}2
I_{(s^0;x^0)}(x^1)\,&{}=\!\sum_{i_1\<, \dots,\> i_r=0}^{k-1}
&& w(s_1^0+i_1a, \dots, s_r^0+i_ra; x^1)\,
\mu_{(s^0; x^0),(s_1^0 + i_1a,\dots,s_r^0 + i_ra; x^1)}\,,
\\[3pt]
I_{(s^0_1+a; x^0)}(x^1)\,&{}=\!\sum_{j_1\<, \dots,\> j_r=0}^{k-1}
&& w(s_1^0+(j_1+1)a, \dots, s_r^0+j_ra; x^1)
\\
&&&\;{}\times\mu_{(s^0_1+a; x^0),(s_1^0 + (j_1+1)a,\dots,s_r^0 + j_ra; x^1)}\,.
\end{alignat*}
If $i_1>0$, then a term
\be
T_1\>=\,w(s_1^0+i_1a, \dots, s_r^0+i_ra; x^1)\,
\mu_{(s^0; x^0),(s_1^0 + i_1a,\dots, s_r^0 + i_ra; x^1)}
\ee
in the first sum has a related term
\be
T_2\>=\,w(s_1^0+i_1a, \dots, s_r^0+i_ra; x^1)\,
\mu_{(s^0_1+a; x^0),( s_1^0 + i_1a,\dots, s_r^0 + i_ra; x^1)}
\ee
in the second sum. Clearly, $T_1=\phi_1(s^0;x^0) \,T_2$\,.

\vsk.3>
If $i_1=0$, then a term
\be
T_1\>=\,w(s_1^0, s_2^0+i_2a, \dots, s_r^0+i_ra; x^1)\,
\mu_{(s^0; x^0),(s^0_1,s_2^0 + i_2a,\dots,s_r^0 + i_ra; x^1)}\,,
\ee
in the first sum has a related term
\be
T_2\>=\,w(s_1^0+ka, s_2^0+i_2a, \dots, s_r^0+i_ra; x^1)\,
\mu_{(s^0_1+a; x^0),(s_1^0 +ka, s^0_2+i_2 a,\dots,s_r^0 + i_ra; x^1)}\,,
\ee
in the second sum. Due to $b$-periodicity,
\bea
T_2\>=\,w(s_1^0, s_2^0+i_2a, \dots, s_r^0+i_ra; x^1)\,
\mu_{(s^0_1+a; x^0),(s_1^0+ka, s^0_2+i_2 a,\dots,s_r^0 + i_ra; x^1)}\,.
\eea
Hence \,$\tilde\phi_1(s^0;x^0)\,T_1 = \phi_1(s^0;x^0) \,T_2$\,.
Therefore, \,$T_1 = \phi_1(s^0;x^0) \,T_2$ \,due to the Bethe ansatz equations
\eqref{bae}\>.

\vsk.3>
The correspondence $T_1\to T_2$ identifies the summands in the first and
second sums and proves the proportionality \eqref{any m} for $l=1$.
Theorem \ref{thm ind} is proved.
\end{proof}

\subsection{Cyclotomic Hamiltonians}
\label{sec cycl}
Suppose that the connections \,$\bigl(L_m(x)\bigr)$\,, and
\\
\,$\bigl(\phi_l(s;x)\,,\psi_m(s;x)\bigr)$ \,and the functions \,$w(s;x)$\,, \,$g_{l,m}(s;x)$
depend on an additional parameter \,$\eps$ \,and are regular at \,$\eps=0$\,.
More precisely, assume the following asymptotics as \,$\eps\to 0$\,:
\begin{align}
\notag
\\[-16pt]
\label{eps to 0}
L_m(x)\,&{}=\,L_m^\0(x)\>\bigl(1+\eps\>X_m(x)+o(\eps)\bigr)\,,
\\[2pt]
\notag
\phi_l(s;x)\,&{}=\,\phi_l^\0(s;x)\>\bigl(1+\eps\>\xi_l(s;x)+o(\eps)\bigr)\,,
\\[2pt]
\notag
\psi_m(s;x)\,&{}=\,\psi_m^\0(s;x)\>\bigl(1+\eps\>\zeta_m(s;x)+o(\eps)\bigr)\,,
\\[-17pt]
\notag
\end{align}
\be
w(s;x)\,=\,w^\0(s;x)+o(1)\,,\qquad g_{l,m}(s;x)\,=\,g_{l,m}^\0(s;x)+o(1)\,.
\ee
These asymptotics imply that
\be
\mu_{(s^0;x^0),(s^1;x^1)}\,=\,\mu^\0_{(s^0;x^0),(s^1;x^1)}+o(1)
\vvn-.8>
\ee
and
\begin{align}
\notag
\\[-22pt]
\label{ix10}
I(x^1)\,&{}=\,I^\0(x^1)+o(1)\,,
\\[4pt]
\notag
I^\0(x^1)\,&{}=\!\sum_{j_1, \dots, j_r=0}^{k-1}
w^\0(s_1^0+j_1a, \dots, s_r^0+j_ra; x^1)\;
\mu^\0_{(s^0; x^0),(s_1^0 + j_1a,\dots, s_r^0 + j_ra; x^1)}\,.
\end{align}

Assume also that the monodromy functions of the connections \,$\bigl(L^\0_m(x)\bigr)$ \,and
\\
\,$\bigl(\phi_l^\0(s;x)\,,\psi_m^\0(s;x)\bigr)$ \,are identities,
\begin{gather*}
\\[-24pt]
L^\0_m\bigl(x_m+(k-1)\>a\bigr)\dots\<\>L_m^\0(x_m+a)L_m^\0(x)\,=\,1\,,
\\[3pt]
\prod_{j=0}^{k-1}\,\phi_l^\0(s_l+j\<\>a;x)\,=\,1\,,\qquad
\prod_{j=0}^{k-1}\,\psi_m^\0(s;x_m+j\<\>a)\,=\,1\,.
\end{gather*}
Then the asymptotics of the monodromy functions as \,$\eps\to0$ \,are
\begin{gather}
\label{Xt def}
\tilde L_m(x)\,=\,1+\eps\>\tilde X_m(x)+o(\eps)\,,
\\[3pt]
\notag
\tilde\phi_l(s;x)\,=\,1+\eps\>\tilde\xi_l(s;x)+o(\eps)\,,\qquad
\tilde\psi_m(s;x)\,=\,1+\eps\>\tilde\zeta_m(s;x)+o(\eps)\,,
\end{gather}
where
\begin{align}
\notag
\\[-30pt]
\label{tilde X}
\tilde X_m(x)\,&{}=\,\sum_{j=0}^{k-1}\,
\bigl(L_m^{(j)}(x)\bigr)^{-1}X_m(x_m+j\<\>a)\,L_m^{(j)}(x)\,,
\\[4pt]
\notag
L_m^{(j)}(x)\,&{}=\,L_m^\0\bigl(x_m+(j-1)\>a\bigr)\dots\<\>L_m^\0(x_m+a)\>L_m^\0(x)\,,
\\[4pt]
\notag
\tilde\xi_l(s;x)\,&{}=\,\sum_{j=0}^{k-1}\,\xi_l(s_l+j\<\>a;x)\,,\qquad
\tilde\zeta_m(s;x)\,=\,\sum_{j=0}^{k-1}\,\zeta_m(s;x_m+j\<\>a)\,.
\end{align}
Notice that the functions \,$\tilde\xi_l(s;x)$ \,and \,$\tilde\zeta_m(s;x)$
\,are \,$a$-periodic with respect to the variables \,$s$ \,and \,$x$\,.

\smallskip
We call matrices \,$\tilde X_1(x),\dots,\tilde X_n(x)$ \,the {\it cyclotomic Hamiltonians\/}.
They pairwise commute by equality \eqref{LL=LL} and asymptotics \eqref{Xt def},
\vvn.4>
\be
\tilde X_l(x)\,\tilde X_m(x)\,=\,\tilde X_m(x)\,\tilde X_l(x)\,,\qquad
1\le l,m\le n\,.
\ee

\medskip
Taking the limit \,$\eps\to 0$ \,in Theorems \ref{thm sol0} and
\ref{thm i sol0} gives respectively Theorems \ref{thm sol00} and
\ref{thm i sol00} below.

\begin{thm}
\label{thm sol00}
For \,$x^1\in \mc O_{x^0}$ \,and \;$m=1,\dots, n$\,,
\vvn.4>
\be
I^\0(x^1_m+a)\,=\,L_m^\0(x^1)\,I^\0(x^1)\,.
\ee
\end{thm}

\smallskip
\begin{thm}
\label{thm i sol00}
For \,$x^1\in \mc O_{x^0}$ \,and \;$m=1,\dots, n$\,,
\be
\tilde X_m(x^1)\,I^\0(x^1)\,=\,\tilde\zeta_m(s^0;x^0)\,I^\0(x^1)\,-\>
\sum_{l=1}^r\,\tilde\xi_l(s^0;x^0)\,\tilde g_{l,m}^{\>\0}(s^0;x^1)\,,
\vv.1>
\ee
where
\;$\tilde g_{l,m}^{\>\0}(s;x)=\>\lim_{\>\eps\to0}\>\tilde g_{l,m}(s;x)$ \>and
the functions \,$\tilde g_{l,m}(s;x)$ are defined by \eqref{glmt}.
\end{thm}

\smallskip
The Bethe ansatz equations for cyclotomic Hamiltonians are the system of equations
\vvn.4>
\beq
\label{bae0}
\tilde \xi_l(s;x) = 0, \qquad l=1,\dots,r.
\vv.4>
\eeq
If $(s^0; x^0)$ is a solution to system \eqref{bae0}, then all points of
the orbit $\mc O_{(s^0;x^0)}$ are also solutions to the Bethe ansatz equations
since the functions \,$\tilde\xi_l(s;x)$ are constant on the orbit \,$\Oc_{(s^0;x^0)}$\,.

\begin{cor}
\label{cor eigen1}
Let $(s^0;x^0)$ be a solution to the Bethe ansatz equations \eqref{bae0}\,.
Then
\vvn.4>
\beq
\label{eXso}
\tilde X_m(x^1)\,I^\0(x^1)\,=\,\tilde\zeta_m(s^0;x^0)\,I^\0(x^1)
\vv.1>
\eeq
for \,$x^1\in\mc O_{x^0}$\,, \,$m=1,\dots,n$\,.
\end{cor}
\begin{proof}
The statement follows from Theorem \ref{thm i sol00}.
\end{proof}

\begin{cor}
\label{cor eigen0}
Let $(s^0;x^0)$ be a solution to the Bethe ansatz equations \eqref{bae0}\,.
Then the vector
\vvn.2>
\beq
\label{ix00}
I^\0(x^0)\,=\!\sum_{j_1, \dots, j_r=0}^{k-1}
w^\0(s_1^0+j_1a, \dots, s_r^0+j_ra; x^0)\;
\mu^\0_{(s^0; x^0),(s_1^0 + j_1a,\dots,s_r^0 + j_ra; x^0)}\,,
\eeq
if nonzero, is an eigenvector of the cyclotomic Hamiltonians $(\tilde X_m(x^0))$
with respective eigenvalues
$(\tilde\zeta_m(s^0;x^0))$\,.
\end{cor}

\subsection{Periodic solutions}
\label{sec period}
Suppose there is a function \,$F(s;x)$ \,that is \,$b$-periodic with respect
to variables \,$s$ \,and \,$x$\,, and such that
\begin{align}
\notag
\\[-16pt]
\label{master}
\phi_l(s;x)\,&{}=\;\frac{F(s_l+a;x)}{F(s;x)}\;,\qquad l=1,\dots,r,
\\[4pt]
\notag
\psi_m(s;x)\,&{}=\;\frac{F(s;x_m+a)}{F(s;x)}\;,\qquad m=1,\dots,n\,.
\end{align}
The function \,$F(s;x)$ \,is the {\it master function\/} of the rank-one
connection
\\
\,$\bigl(\phi_l(s;x)\,,\psi_m(s;x)\bigr)$\,.
Equalities \eqref{master} imply that the monodromy functions
\,$\tilde\phi_l(s;x)\,,\tilde\psi_m(s;x)$ \,are identities,
\vvn.4>
\be
\tilde\phi_l(s;x)\>=\>1\,,\quad l=1,\dots,r,\qquad
\tilde\psi_m(s;x)\>=\>1\,,\quad m=1,\dots,n\,.
\vv.4>
\ee
Furthermore, \;$\mu_{(s^0; x^0),(s^1; x^1)}=F(s^1;x^1)/F(s^0;x^0)$ \,and
\vvn.4>
\be
I(x^1)\,=\,J(s^0;x^1)/F(s^0;x^0)\,,
\vv-.5>
\ee
where
\vvn-.5>
\be
J(s;x)\,=\!\sum_{j_1, \dots, j_r=0}^{k-1}
w(s_1+j_1a, \dots, s_r+j_ra; x)\,F(s_1 + j_1a,\dots, s_r + j_ra; x)\,.
\vv.2>
\ee

\begin{thm}
\label{thm solJ}

We have
\bean
\label{ilib}
J(s;x_m+b)\,&=&\,J(s;x)\,,\qquad m=1,\dots,n\,,
\\
\label{iliJ}
J(s;x_m+a)\,&=&\,L_m(x)\,J(s;x)\,,\qquad m=1,\dots,n\,,
\\
\label{ILIJ}
J(s_l+a;x)\,&=& \,J(s;x)\,,\qquad l=1,\dots,r\,.
\eean

\end{thm}

The proof is similar to that of Theorem \ref{thm sol0}.
\smallskip

By Theorem \ref{thm solJ}, the $V$-valued function $J(s;x)$ is a $b$-periodic flat section of the discrete connection
$(L_m(x))$ for any fixed value of $s$.

\begin{cor}
\label{thm i solJ}

For any fixed value of $s$, the flat section \,$J(s;x)$ of Theorem \ref{thm solJ} is an eigensection
of the monodromy functions with the eigenvalue \,$1$\,, that is,
\vvn.4>
\beq
\label{e soJ}
\tilde L_m(x)\,J(s;x)\,=\,J(s;x)\,,\qquad m=1,\dots,n\,.
\vv.2>
\eeq
\end{cor}
\begin{proof}
By Theorem \ref{thm solJ},
\vvn.3>
\be
\tilde L_m(x)\,J(s;x)\,=\,J(s;x_m+k\<\>a)\,=\,J(s;x)\,,
\vv.3>
\ee
since \,$J(s;x)$ \,is \,$b$-periodic with respect to variables \,$x$\,.
\end{proof}

\begin{rem}

Given a $k$-resonant $b$-periodic discrete flat connection $(L_m(x))$ with
step $a$, one can construct flat eigensections of the monodromy functions of this connection whenever
an integral representation of the type described in Section \ref{sec 3} is
available. Such integral representations are known to exist for the systems of
trigonometric qKZ difference equations and for the systems of quantum
difference equations for Nakajima varieties. If the multiplicative step of
a system is a root of unity, then the construction of Section \ref{sec 3}
applies, and the corresponding flat eigensections can be obtained.

To keep the notation minimal, in this paper we focus on applications of the general construction
to the system of trigonometric qKZ difference equations associated with a tensor product of Verma modules over $U_q(\mathfrak{sl}_2)$ with multiplicative root-of-unity step.

\end{rem}

\section{Discrete qKZ connection}
\label{sec 4}

\subsection{Tensor product of Verma modules}

For $q\in\C^\times$, $q^2\<\<\ne 1$, consider the quantum group $U_q(\frak{sl}_2)$ generated
by the elements \,$e,f,q^{\pm h}$ subject to relations
\be
q^hq^{-h}=q^{-h}q^h=1,\quad
q^{\pm h}e=q^{\pm2}e\>q^{\pm h},\quad
q^{\pm h}\<f=q^{\mp2}f\<\>q^{\pm h},\quad
[\<\>e,f\<\>]=\frac{q^h-q^{-h}}{q-q^{-1}}\,.
\ee
Let $M_\La$ denote the Verma module over $U_q(\frak{sl}_2)$
with highest weight $\La\in\C$ and highest weight vector $v_\La$:
\be
e\<\>v_\La=0,\qquad q^{\pm h}v_\La=q^{\pm\La}v_\La.
\ee
A basis of $M_\La$ is formed by the vectors $f^{\<\>r}v_\La$, $r\in\Z_{\geq 0}$.

We have the weight decomposition
\be
{\tsize\bigotimes}^n_{j=1} M_{{\La_j}} =
{\tsize\bigoplus}_{r=0}^\infty
\left({\tsize\bigotimes}^n_{j=1} M_{{\La_j}}\right)\left[\sum_{j=1}^n\La_j-2r\right].
\ee
The basis of a weight subspace
$V=\left(\bigotimes^n_{j=1} M_{{\La_j}}\right)\left[\sum_{j=1}^n\La_j-2r\right]$
is formed by the vectors
\beq
\label{basis}
f^{\<\>\vec r}:= f^{\<\>r_1}v_{\La_1} \ox\dots\ox f^{\<\>r_n}v_{\La_n}\,,
\eeq
where $\vec r=(r_1,\dots,r_n)$, $r_1+\dots+r_n = r$.
The set of such indices $\vec r$ is denoted by $\mc I_r$.

\subsection{Trigonometric $R$-matrix}
\label{R-m}
The trigonometric $R$-matrix $R_{\La_1, \La_2}(x)\in \End M_{\La_1}\ox M_{\La_2}$
is defined by the relations
\begin{align*}
R_{\La_1, \La_2}(x)\>(\<\>f\ox 1+q^h\ox f\<\>)\,&{}=\,
(\<\>f \ox q^h+1\ox f\<\>)\>R_{\La_1, \La_2}(x)\,,
\\[4pt]
R_{\La_1, \La_2}(x)\>(\<\>f\ox q^h+x\<\<\cdot\!1\ox f\<\>)\,&{}=\,
(\<\>f \ox 1+x\<\<\cdot\<q^h\ox f\<\>)\>R_{\La_1, \La_2}(x)\,,
\\[4pt]
R_{\La_1, \La_2}(x)\>v_{\La_1}\ox v_{\La_2}&{}=\,v_{\La_1}\ox v_{\La_2}\,,
\end{align*}
see, for example, \cite{TV2}.\footnote{\strut\,%
The generators of $U_q(\frak{sl}_2)$ here and in \cite{TV2} are related
as follows: $e =q^{-H}E\>,\;f \<\>=Fq^H,\;q^{\pm h}\<\<=q^{\pm2H}$.
The relation of highest weights is
\,$\La_{\mathrm{here}}\<=2\La_{\mathrm{[TV2]}}$\>.}
The $R$-matrix preserves the weight decomposition
of $M_{{\La_1}}\ox M_{{\La_2}}$.

The restrictions of \,$R_{\La_1, \La_2}(x)$ \,and
\,$\bigl(R_{\La_1,\La_2}(x)\bigr)^{-1}$ to the weight subspace
$\bigl(M_{\La_1}\!\ox M_{\La_2}\bigr)\<\>[\<\>\La_1\<+\La_2\<-2\<\>r\<\>]$
have respectively the form
\be
P(x,q^2\<\<,q^{\La_1}\<\<,q^{\La_2})\,
\prod_{s=0}^{r-1}\,(q^{\La_1\<+\La_2}\<-q^{2s}\<x\<\>)^{-1}
\quad\text{and}\quad\,
Q(x,q^2\<\<,q^{\La_1}\<\<,q^{\La_2})\,
\prod_{s=0}^{r-1}\,(q^{\La_1\<+\La_2}x-q^{2s}\<\>)^{-1}\>,
\ee
where \,$P(x,q^2\<\<,q^{\La_1}\<\<,q^{\La_2})$ \,and
\,$Q(x,q^2\<\<,q^{\La_1}\<\<,q^{\La_2})$ \,are polynomials
in \,$x,q^2\<\<,q^{\La_1}\<\<,q^{\La_2}$ with integer coefficients.
See Motivating Example \ref{n=2 p=-1}.

\smallskip

\subsection{Discrete multiplicative qKZ connection}
\label{sec 3.4}

For $p, \ka\in\C^\times$, define the
$\End \bigotimes^n_{j=1} M_{{\La_j}}$-valued qKZ functions:
\bea
K_m(z):=K_m(z_1,\dots,z_n, p, q, q^{\La_1},\dots, q^{\La_n}, \ka),
\eea
$m=1,\dots,n$,
by the formula:
\bean
\label{qkz op}
K_m(z)
&=&
R_{{\La_{m}}, {\La_{m-1}}}^{(m,m-1)}\left(\frac{pz_m}{z_{m-1}}\right)
\dots
R_{{\La_{m}}, {\La_{1}}}^{(m,1)}\left(\frac{pz_m}{z_{1}}\right)
\ka^{(\La_m - h^{(m)}\<)\</2}\times{}
\\
\notag
&\times &
R_{{\La_{m}}, {\La_{n}}}^{(m,n)}\left(\frac{z_m}{z_n}\right) \dots
R_{{\La_{m}}, {\La_{m+1}}}^{(m,m+1)}\left(\frac{z_m}{z_{m+1}}\right).
\eean
Every qKZ function preserves the weight decomposition of
$\bigotimes^n_{j=1} M_{{\La_j}}$\,.

\smallskip
The qKZ functions $(K_m(z))$ define a flat multiplicative discrete connection
on the trivial bundle $\ox^n_{j=1} M_{{\La_j}}\times \C^n\to \C^n$, that is,
\bean
\label{f qKZ}
K_m(pz_l) K_l(z) = K_l(pz_m) K_m(z), \qquad 1\leq l,m\leq n.
\eean
The number $p$ is called the multiplicative step of the discrete qKZ
connection.\footnote{\strut\,%
Here $K_m(pz_l) $ denotes $ K_m(z_1,\dots, z_{l-1}, pz_l,z_{l+1},\dots, z_n)$
and
$K_l(pz_m) $ denotes $ K_l(z_1$, $\dots$, $z_{m-1},$ $pz_m$, $z_{m+1}$, \dots,
$ z_n)$ }

\medskip

An $\bigotimes^n_{j=1} M_{{\La_j}}$-valued function $I(z)$ is a flat section of this connection if
\beq
\label{flat qkz}
I( pz_m) = K_m(z) I(z), \qquad m=1,\dots,n.
\eeq
This system of equations is called the trigonometric qKZ equations.

\subsection{Monodromy of resonance}

Let $k>1$ be an integer.
The discrete qKZ connection $(K_m(z))$
is $k$-resonant if
the multiplicative step is a primitive root of unity of order $k$, that is,
\bea
p = e^{2\pi i j/k}, \quad 1\leq j<k, \quad (j,k)=1.
\eea
Then the monodromy functions are defined as:\bean
\label{K def}
\tilde K_m(z)
&=&
K_m(p^{k-1}z_m) \dots K_m(pz_m ) K_m(z), \qquad
m=1,\dots, n.
\eean

\begin{lem}

For $1\leq l,m\leq n$, we have:
\bean
\label{Klm=ml}
\tilde K_l(pz_m) K_m(z)
& =&
K_m(z)\tilde K_l(z),
\eean
and the monodromy functions commute,
\beq
\label{KLL=LL}
\tilde K_l(z) \tilde K_m(z) = \tilde K_m(z) \tilde K_l(z),
\qquad 1\leq l,m\leq n.
\eeq
\end{lem}

\begin{proof}
The proof follows from flatness relations \eqref{f qKZ}.
\end{proof}

\begin{cor}
If $I(z)$ is a flat section,
then $\tilde K_l(z) I(z)$, $l=1,\dots,n$, are flat sections.
\end{cor}

The group $\Z^{n}$ acts on the base of the qKZ connection:
\bean
\label{acts}
\left(\left(z_1^0,\dots, z_n^0\right),(r_1,\dots,r_n)\right) \mapsto
\left(p^{r_1}z_1^0,\dots, p^{r_n}z_n^0\right).
\eean
For a point $z^0$ of the base, let $\mc O_{z^0}$ denote
the orbit of $z^0$.
The discrete connection can be restricted to $\mc O_{z^0}$ and then lifted from $\mc O_{z^0}$ to $\Z^{n}$.
Let $\Ga(\mc O_{z^0})$ be the vector space of flat sections
of the lifting from $\mc O_{z^0}$ to $\Z^{n}$.
The map $\Ga(\mc O_{z^0})\to V$, $I \mapsto I(0)$, is an isomorphism.
Here $0$ is the zero element of $\Z^{n}$.
The monodromy functions $(\tilde K_l(z))$ restricted to $\Ga(\mc O_{z^0})$ provide
$\Ga(\mc O_{z^0})$ with a collection of commuting linear operators.

\smallskip

Our goal is to construct joint eigenvectors of the monodromy functions.

\section{Multiplicative integral representation}
\label{sec 5}

\subsection{Multiplicative connection of rank 1}

Fix a weight subspace
\bea
V=\left(\bigotimes^n_{j=1} M_{{\La_j}}\right)\left[\sum_{j=1}^n\La_j-2r\right].
\eea
Let $p\in\C^\times$. Let
$t=(t_1,\dots, t_r)$, $z=(z_1,\dots,z_n)$ be variables.
A multiplicative discrete flat connection with step $p$ on the space with coordinates
$(t;z)$ is a collection of nonzero scalar transition functions
$\Phi_l(t;z)$, $l=1,\dots,r,$ and $\Psi_m(t;z),$ $m=1,\dots,n$, such that
\bean
\label{flats}
\Phi_l(p\<\>t_i;z) \,\Phi_i(t;z)
&=&
\Phi_i(p\<\>t_l;z) \,\Phi_l(t;z), \qquad
\ \
1\leq l,i\leq r,
\\
\notag
\Psi_m(t; pz_j) \,\Psi_j(t;z)
&=&
\Psi_j(t; pz_m) \,\Psi_m(t;z), \qquad 1\leq m,j\leq n,
\\
\notag
\Psi_m(p\<\>t_i; z) \,\Phi_i(t;z)
&=&
\Phi_i(t; pz_m) \,\Psi_m(t;z), \qquad 1\leq i\leq r, \quad
1\leq m\leq n.
\eean

Let $k>1$ be an integer.
The discrete connection of rank 1 is $k$-resonant if
the multiplicative step is a primitive root of unity of order $k$.
In that case, the monodromy functions are defined as:
\be
\tilde \Phi_l(t;z)=
\prod_{j=0}^{k-1} \Phi_l(p^jt_l;z),\qquad
\tilde \Psi_m(t;z)=
\prod_{j=0}^{k-1} \Psi_m(t; p^jz_m),
\ee
$l=1,\dots, r$, $m=1,\dots, n$.

\begin{lem}
\label{lem 5.1}
The monodromy functions are $p$-periodic in variables $t$ and $z$,
\bea
\tilde \Phi_l(p\<\>t_i;z)
&=&
\tilde \Phi_l(t;z),
\qquad \ \ \,
\tilde \Phi_l(t;pz_j)
=
\tilde \Phi_l(t;z),
\\
\tilde \Psi_m(p\<\>t_i;z)
&=&
\tilde \Psi_m(t;z),
\qquad
\tilde \Psi_m(t;pz_j)
=
\tilde \Psi_m(t;z).
\eea

\end{lem}

The proof is similar to the proof of Lemma \ref{lem 2.1}.

\smallskip

The system of Bethe ansatz equations is the system:
\bean
\label{m BAE}
\tilde \Phi_l(t;z) = 1,\qquad l=1, \dots, r.
\eean

\subsection{Connection of master function}

Let
\begin{align}
\notag
\\[-16pt]
\label{Phi}
\Phi_l(t;z)\,&{}=\,
\ka\;\prod_{i=1}^n\,\frac{1-q^{\La_i}t_l/z_i}{q^{\Lambda_i}-p\<\>t_l/z_i}\,\cdot
\prod_{j\ne l}\,\frac{(1-p\<\>t_l/t_j)(q^2\<\<-t_l/t_j)}{(1-t_l/t_j)(1-p\<\>q^2t_l/t_j)}\;,
\\[3pt]
\label{Psi}
\Psi_m(t;z)\,&{}=\,\prod_{j=1}^r\,\frac{q^{\Lambda_m}-t_j/z_m}{1-q^{\La_m}t_j/(pz_m)}\;,
\end{align}
$l=1,\dots,r$\,, \,$m=1,\dots,n$\,.

\begin{lem}

For $p\in\C^\times$, these functions
define
a multiplicative discrete flat connection of rank 1.

\end{lem}

\smallskip
The flatness conditions \eqref{flats} can be checked directly. They also follow from the existence of a (master)
function $F(t;z)$ such that
$$\Phi_l(t;z) = \frac{F(p\<\>t_l;z)}{F(t;z)},
\qquad
\Psi_m(t;z) = \frac{F(t;pz_m)}{F(t;z)},$$
for \,$l=1,\dots,r$\>, and \,$m=1,\dots,n$\,.
Formulas for the master function can be found in \cite{TV1,TV2}. The connection defined by
\eqref{Phi} and \eqref{Psi} is called the multiplicative connection of master function.

\begin{lem}
\label{lem 7.3}
Let $p =e^{2\pi i j/k}$, \,$1\leq j<k$, \,$(j,k)=1$. Then
the monodromy
functions of the multiplicative connection of master function
take the form:
\begin{align}
\notag
\\[-15pt]
\label{mon phi no}
\tilde\Phi_l(t;z)\,&{}=\,
\ka^k\,\prod_{i=1}^n\,\frac{1-q^{k\La_i}\<\>t_l^k/z_i^k}{q^{k\Lambda_i}-t_l^k/z_i^k}
\,\cdot\prod_{j\ne l}\,\frac{q^{2k}\<\<-t_l^k/t_j^k}{1-q^{2k}\<\>t_l^k/t_j^k}\;,
\\[4pt]
\label{mon psi no}
\tilde\Psi_m(t;z)\,&{}=\,
\prod_{j=1}^r\,\frac{q^{k\Lambda_m}-t_j^k/z_m^k}{1-q^{k\La_m}\<\>t_j^k/z_m^k}\;,
\end{align}
$l=1,\dots,r$, \;$m=1,\dots,n$\,.
\end{lem}

\begin{proof}
The lemma follows from the formula
\,$\prod_{j=0}^{k-1} (A - p^j B) = A^k\< - B^k$.
\end{proof}

\subsection{Weight function}

The $V$-valued weight function
\beq
\label{WF}
W(t;z)\,=\>\sum_{\vec r\in \mc I_r} W_{\<\vec r\>}(t;z)\,f^{\<\>\vec r}
\eeq
is defined by the formula
\bean
\label{4.8}
&&
W_{\<\vec r\>}(t;z)\,=\,
\prod_{i=1}^n\,
\frac{q^{r_i(r_i-1)}}{[r_i]_{q^2}!} \>\times{}
\\
\notag
&&
\times \ \
\on{Sym}_{t_1,\dots,t_r}\
\left(\prod_{i=1}^n
\left(\prod_{j=\vec r^{(i-1)}+1}^{\vec r^{(i)}}
\left((t_j/z_i)\prod_{i'<i} \left(1-q^{\La_{i'}} t_j/z_{i'}\right)
\prod_{i''>i}\left(q^{\La_{i''}}-t_j/z_{i''}\right)\right) \right)\right.
\\
\notag
&&
\phantom{aaaaaaaaaaaaa}
\times \ \left.
\prod_{j'<j''}
\frac{1-q^2t_{j'}/t_{j''}}
{1-t_{j'}/t_{j''}} \right).
\eean
Here we use the following notation. For a function $f(t;z)$, define
\be
\on{Sym}_{t_1,\dots,t_r}f(t;z) =
\sum_{\si\in S_r} f(t_{\si(1)}, \dots,t_{\si(r)}; z)\,.
\ee
For a positive integer $m$, denote
\be
[m]_{q^2}= \frac{1-q^{2m}}{1-q^2}\quad\text{and}\quad
[m]_{q^2}! = \prod_{j=1}^m[j]_{q^2}.
\ee
For a vector $\vec r = (r_1,\dots,r_n)$\,, \;$r_1+\dots+r_n = r$, set
\,$\vec r^{(0)}\!=0\,$ and \,$\vec r^{(i)}\!=r_1+\dots + r_i$\,, \;$i=1,\dots n$\,.

\medskip

For every $\vec r\in \mc I_r$\,, the function
\bean
\label{poly}
z_1^r\<\cdots z_n^r\,W_{\<\vec r\>}(t;z)
\eean
is a polynomial in $t, z, q$ and $q^{\La_1},\dots, q^{\La_n}$.
The factors $[r_i]_{q^2}!$ and $1-t_{j'}/t_{j''}$ in the denominator of formula \eqref{4.8}
are canceled out by the symmetrization. See \cite{TV2}.

\subsection{Theorem on integral representation}

\begin{thm} [\cite{TV1, TV2}]
\label{thm 7.4}
Given a nonzero complex number $p$,
there exist rational $V$-valued functions
$G_{l,m}(t;z)$, $l=1, \dots, r$, $m=1,\dots,n$, such that
\beq
\label{mIR}
\phantom{aaa}
\Psi_m(t;z)\<\>W(t;pz_m)\,=\,K_m(z)\<\>W(t;z)\>+
\sum_{l=1}^r \left(\Phi_l(t;z)\>G_{l,m}(p\<\>t_l;z) -G_{l,m}(t;z)\right),
\eeq
$m=1,\dots,n$.
Furthermore, the coordinate functions of every $G_{l,m}(t;z)$ in the basis $(f^{\vec r})$
are ratios
$\frac{P(t;z)}{Q(t;z)}$\,,
where $P(t;z)$ is a Laurent polynomial in $t, z, p, q, q^{\La_1},\dots,q^{\La_n}, \ka $
with integer coefficients,
and $Q(t;z)$ is the product with all factors of the form
\begin{align*}
q^{\Lambda_j}-p^ct_i/z_j\,, &\quad i = 1,\dots, r, \ j=1,\dots,n,
\\
1-p^cq^{\La_j}t_i/z_j\,, &\quad i = 1,\dots, r, \ j=1,\dots,n,
\\
1-p^c\<\>t_i/t_j \,, & \quad i,j = 1,\dots, r, \ i\ne j,
\\
1-p^cq^{2}\<\>t_i/t_j\,, &\quad i,j = 1,\dots, r, \ i\ne j,
\\
q^{\La_i+\La_j-2s} - p^c z_i/z_j, &\quad
i,j=1,\dots,n, \ i\ne j, \ \ s=0,\dots,r-1,
\end{align*}
where $c\in\Z$.\footnote{${}$ These factors are
the factors of the denominators of the functions
$(\Phi_l, \Psi_m)$ in which variables $t$ and $z$ are shifted by powers of $p$.}

\end{thm}

A system of such relations is called a {\it multiplicative integral representation} for
the qKZ connection $\bigl(K_m(z)\bigr)$\,.

\subsection{Logarithmic coordinates}
\label{sec 3.6}

Let
\be
p=e^{2\pi i j/k}, \qquad 1\leq j< k, \quad (j,k)=1.
\ee
Fix numbers $\ga$ and $\delta, $ such that
\beq
\label{log coo}
q= e^{2\pi i \ga}, \
\ka = e^{2\pi i \delta}.
\eeq
Introduce variables $s=(s_1,\dots,s_r)$ and $x=(x_1,\dots,x_n)$ by the formulas:
\bean
\label{log tz}
t_l = e^{2\pi i s_l}, \qquad
z_m = e^{2\pi i x_m}.
\eean
Define the
$\on{GL}(V)$-valued functions
\bea
L_m(x) := L_m
(x_1,\dots, x_n, \ga, \La_1,\dots,\La_n, \ka),
\eea
$m=1,\dots,n$,
by the formulas:
\beq
\label{qkz log}
L_m(x)\,=\,
K_m(e^{2\pi i x_1},\dots, e^{2\pi i x_n}, e^{2\pi i j/k}, e^{2\pi i \ga},
e^{2\pi i \ga \La_1}, \dots, e^{2\pi i \ga \La_n}, e^{2\pi i \delta}).
\eeq
Denote
\be
a=\frac jk, \qquad b= j.
\ee
Then $b=ka$. The flatness conditions \eqref{f qKZ} become:
\beq
\label{f qkz}
L_m(x_l + a)\,L_l(x)\,=\,L_l(x_m + a)\,L_m(x)\,,
\eeq
$ 1\leq l,m\leq n$.
We also have the periodicity properties:
\bean
\label{p qkz}
L_m(x_l +b)=L_m(x), \qquad 1\leq l,m\leq n.
\eean
{\bf Summary.} The operators $(L_m
(x))$ form a $k$-resonant
$b$-periodic discrete flat connection with step $a$ in the sense of Section \ref{sec 2}.
We call it the additive qKZ connection.

\medskip

Let
\beq
\label{l1}
w(s;x)\,=\sum_{\vec r}w_{\vec r}\>f^{\<\>\vec r},\quad g_{l,m}(s;x)\,,
\quad \phi_l(s;x)\,,\quad \psi_m(s;x)\,,\quad\tilde \phi_l(s;x)\,,\quad
\tilde\psi_m(s;x)
\eeq
be the functions obtained from the functions
\be
W(t;z)=\sum_{\vec r}^{\vp1} W_{\vec r}\>f^{\<\>\vec r},\quad G_{l,m}(t;z)\,,\quad
\Phi_l(t;z)\,,\quad \Psi_m(t;z)\,,\quad \tilde \Phi_l(t;z)\,,\quad \tilde\Psi_m(t;z)
\ee
of Theorem \ref{thm 7.4} by replacing in them
all
\bea
t_u,\quad z_v, \quad p,\quad q, \quad q^{\La_v},\quad \ka
\eea
with
\bea
e^{2\pi i s_u}, \quad e^{2\pi i x_v},
\quad e^{2\pi i j/k}, \quad e^{2\pi i \ga},\quad
e^{2\pi i \ga \La_v}, \quad e^{2\pi i \delta},
\eea
respectively.

\begin{lem}
\label{lem int rep}

The functions
$w(s;x)$, $g_{l,m}(s;x)$, $\phi_l(s;x)$, $\psi_m(s;x)$ are $b$-periodic with respect to variables
$s$ and $x$ and define an integral representation of
the $k$-resonant $b$-periodic discrete flat additive qKZ connection
$(L_m(x))$ with step $a$,
\bean
\label{IR 22}
\phantom{aaaaaa}
\psi_m(s;x) w(s;x_m+a) = L_m(x)w(s;x) +
\sum_{l=1}^r \left(\phi_l(s;x)g_{l,m}(s_l+a;x) -g_{l,m}(s;x)\right),
\eean
$m=1,\dots,n$.
\qed
\end{lem}

\begin{cor}
\label{cor 5.6}
The statements of Theorems \>\ref{thm sol0}\,--\,\ref{thm i sol},
Corollary \>\ref{cor eigen}, and Theorem \ref{thm ind} \,hold for the
\>$k$-resonant \,$b$-periodic discrete additive qKZ connection \>$(L_m(x))$.
\end{cor}

As an example, we formulate Corollary \ref{cor eigen} for the additive qKZ connection
$(L_m(x))$ in terms of the original multiplicative qKZ functions $(K_m(z))$.

\begin{cor}
\label{cor m eigen}

Let $p=e^{2\pi i j/k}$, $1\leq j<k$, $(j,k)=1$.
Let $(t^0; z^0)$ be a solution to the system of Bethe ansatz equations,
\beq
\label{m bae}
\kappa^k\>\prod_{i=1}^n\,\frac{1-q^{k\La_i}\<\>t_l^k/z_i^k}{q^{k\Lambda_i}-t_l^k/z_i^k}
\,\cdot
\prod_{l'\<\ne\>l}\,\frac{q^{2k}\<\<-t_l^k/t_{l'}^k}{1-q^{2k}\<\>t_l^k/t_{l'}^k}\;=\,1\,,
\eeq
$l=1,\dots,r$. Define a vector in $V$,
\bean
\label{mx0}
I(t^0;z^0)
&=&
\sum_{c_1\<, \dots,\>c_r=0}^{k-1}
W(p^{c_1}t_1^0, \dots, p^{c_r}t_r^0; z^0)
\\
\notag
&&
\phantom{aa}
\times \ \ \prod_{i=1}^r\prod_{d_i=0}^{c_i-1}
\Phi_i (p^{c_1}t_1^0, \dots, p^{c_{i-1}} t^0_{i-1},
p^{d_i}t_i^0, t_{i+1}^0, \dots, t_r^0;z^0)\,.
\eean
If nonzero, $I(t^0;z^0)$ is an eigenvector of the monodromy operators $\tilde K_m(z^0)$,
$m=1,\dots,n$, with respective eigenvalues
\beq
\label{Mon psi}
\prod_{i=1}^r\,\frac{q^{k\Lambda_m}-(t_i^0\<\>/\<z^0_m)^k}
{1-q^{k\La_m}\<\>(t_i^0\<\>/\<z^0_m)^k}\;.
\eeq

Parallel transport of the fiber $V$ over $z^0$ to the fibers over the points of the orbit $\mc O_{z^0}$
extends the vector $I(t^0;z^0)$ to a multi-valued section over $\mc O_{z^0}$ which is an
eigenvector of the
monodromy functions $(\tilde K_m(z))$.
\end{cor}

\section{Laurent-polynomial flat sections}
\label{sec 6}

To illustrate our results, we consider a special case of the discrete qKZ connection \eqref{qkz op} on
\[
V\times\C^n \to \C^n,
\qquad
V=\left(\bigotimes^n_{j=1} M_{\Lambda_j}\right)\left[\sum_{j=1}^n\Lambda_j-2r\right]\!.
\]
Throughout this section, we assume that
\begin{align}
\label{sc1}
& p=e^{2\pi i j/k}, \qquad 1\le j<k, \quad (j,k)=1,
\\
\notag
& q^k=1, \qquad \kappa^k=1, \qquad \Lambda_1,\dots,\Lambda_n\in\mathbb Z.
\end{align}
Under these assumptions, we show that the constructions of Sections~\ref{sec 2}
and~\ref{sec 3} produce Laurent-polynomial global flat sections of the discrete
qKZ connection \eqref{qkz op}.

\subsection{Master function}

For \,$c\in \C$ \,such that \,$c^k=1$\,, let \,$\bar c\in\{\<\>0,\dots,k-1\<\>\}$
\,be such that \,$c=p^{\bar c+1}$. Denote
\be
g(t\<\>;c)\,=\,(\<\>p\<\>t;p)_{\bar c}\>=(1-p\<\>t)\>(1-p^2\<\>t)\dots(1-p^{\bar c}t)\,.
\ee
Then
\be
\frac{g(p\<\>t\<\>;c)}{g(t\<\>;c)}\,=\,\frac{1-c\<\>t}{1-p\<\>t}\;.
\ee

\medskip
\noindent
Let \,$\hat\ka_1\>,\dots,\hat\ka_{\>r}\>,\>\La_1',\dots,\La_n'\!\in\{\<\>0,\dots,k-1\<\>\}$
\,be such that
\be
p^{\<\>\hat\ka_l}=\,\ka\>q^{\<\>4\<\>l-2r-2-\sum_{m=1}^n\La_m},\quad l=1,\dots,r\>,
\qquad p^{\La_m'}=\,q^{r\La_m}\>,\quad m=1,\dots,n\,.
\ee

\begin{lem}
Under assumptions \eqref{sc1}, the function
\be
F(t\<\>;z)\,=\,\prod_{l=1}^r\>t_l^{\<\>\hat\ka_l}\<\cdot\<
\prod_{m=1}^n z_m^{\La_m'}\cdot
\prod_{m=1}^n\,\prod_{l=1}^r\,g(q^{-\La_m}t_l/z_m\>;q^{2\La_m})\,\cdot\<\<
\prod_{l=2}^r\,\prod_{i=1}^{l-1}\,\bigl((1-t_l/t_i)\,g(q^2t_l/t_i\>;q^{-4})\bigr)
\ee
is a Laurent polynomial in \,$t\<$ and \>$z$. Furthermore, the function $F(t\<\>;z)$
is a master function for the connection defined by \eqref{Phi} and
\eqref{Psi}\>, that is,
\be
\Phi_l(t\<\>;z)\,=\,\frac{F(p\<\>t_l\>;z)}{F(t\<\>;z)}\;,\qquad
\Psi_m(t\<\>;z)\,=\,\frac{F(t\<\>;pz_m)}{F(t\<\>;z)}\;,
\ee
for $l=1,\dots,r$ and \,$m=1,\dots, n$.
\end{lem}

\begin{proof}
The proof is by inspection.
\end{proof}

\subsection{Theorem on Laurent polynomials}

The function $F(t;z)\,W(t;z)$ is a $V$-valued Laurent polynomial.
Consider the expansion
\bea
W(t\<\>;z)\>F(t\<\>;z)\,=\!\sum_{l_1,\dots,\>l_r\in\Z}\!
Q_{l_1,\dots,\>l_r}(z)\;t_1^{\<\>l_1}\!\<\dots t_r^{\<\>l_r}\>.
\eea

\begin{thm}
\label{thm Lau}
For any \,$l_1,\dots,l_r\in \Z$, the $V$-valued Laurent polynomial
\,$Q_{kl_1,\dots,\>kl_r}(z)$ is a flat section of the multiplicative qKZ
connection,
\beq
\label{flat lp}
Q_{kl_1,\dots,\>kl_r}( pz_m) = K_m(z)\,Q_{kl_1,\dots,\>kl_r}(z), \qquad m=1,\dots,n.
\eeq
Furthermore, $Q_{kl_1,\dots,\>kl_r}(z)$, if nonzero, is an eigensection of the monodromy operators with
all eigenvalues equal to 1, that is
\be
\tilde K_m(z)\,Q_{kl_1,\dots,\>kl_r}(z) = Q_{kl_1,\dots,\>kl_r}(z).
\ee
\end{thm}

\begin{proof}
Define
\be
J(t;z)\,=\!\sum_{c_1\<,\dots,\>c_r=0}^{k-1}
W(p^{c_1}t_1, \dots, p^{c_r}t_r; z)\,F(p^{c_1}t_1, \dots, p^{c_r}t_r; z)\,.
\ee
Since
\be
\sum_{c=0}^{k-1}\,p^{\>c\<\>m}=\,0 \quad \on{if} \ k\nmid m\,,\quad \on{and}\quad
\;\sum_{c=0}^{k-1}\,p^{\>c\<\>m}=\,k \quad \on{if}\ k\mid\<m\,,
\ee
we have
\be
J(t\<\>;z)\,=\,k\!\sum_{l_1,\dots,l_r\in\Z}\!
Q_{\>l_1k,\dots,\>l_rk}(z)\;t_1^{l_1k}\!\<\dots t_r^{l_rk}\>.
\vv.4>
\ee
Also, by Theorem \ref{thm solJ},
the $V$-valued Laurent polynomial \,$J(t;z)$ \,is a flat section of
the multiplicative qKZ connection,
\be
J(t;pz_m)\,=\,K_m(z)\,J(t;z)\,, \qquad m=1,\dots,n\,.
\ee
These two observations imply the first statement of the theorem.
The second statement follows from Corollary \ref{thm i solJ}.
\end{proof}

Recall that
\be
p=e^{2\pi i j/k}, \qquad 1\le j<k, \quad (j,k)=1.
\ee

\begin{exmp}
Assume that \,$\La_i=1$ for $i=1,\dots,n$\,, and \,$r=1$\,.
Suppose \,$q^2\<=p^m$ \,and \,$\ka=p^{m'}\!\<q^n$ for some
\,$m,m'\!\in\{\<\>0,\dots,k-1\<\>\}$\,, \,$m\ne 0$\,. In that case, the function
\,$F(t\<\>;z)\>W(t\<\>;z)$ \,is a polynomial in \,$t$ \,of degree $m\<\>n+m'$.
Let \,$m\<\>n$ \,be divisible by \,$k$\,. Then Theorem \ref{thm Lau} gives
\,$\frac m k\, n$ \,Laurent-polynomial flat sections, while the fiber of this
discrete connection is \,$n$-dimensional. Probably, these Laurent-polynomial
flat sections span a module of rank \,$\frac m k\, n$.
\vsk.2>

It would be interesting to determine if the number \,$\frac m k\, n$
\,is related to the dimension of the first subspace of some Frobenius
filtration, cf.~\cite{VV, MV2} where $\frac n2$ had this property.
\end{exmp}

\appendix

\section{Bethe ansatz for Cyclotomic Gaudin Model}
\label{appA}

In Sections \ref{sec 4}\,--\,\ref{sec 6}, we constructed eigenvectors and
eigenvalues of the monodromy operators of the qKZ discrete connection
associated with a tensor product of Verma modules over the quantum group
$U_q(\frak{sl}_2)$, in the case when the multiplicative step $p$ is a root of
unity. In this appendix, we describe the limit of this construction as
$q = 1-\eps$ and $\eps \to 0$. We show that the limit of the associated
monodromy operators gives the cyclotomic Gaudin operators associated with the
corresponding tensor product of Verma modules over $\frak{sl}_2$, for the same
root of unity $p$. Furthermore, the limit of our construction yields
eigenvectors and eigenvalues of the cyclotomic Gaudin Hamiltonians.

\subsection{Verma modules over $\frak{sl}_2$}

Consider the complex Lie algebra \,$\frak{sl}_2$ with the standard generators
\,${\bf e, f, h}$ and relations
\be
[{\bf\<\>h, e\<\>]=2\<\>e, \quad [\<\>h,f\<\>] = -2\<\>f, \quad [\<\>e,f\<\>] = h.}
\ee
For $\La \in \C$, let ${\bf M}_\La$ be the Verma module over \,$\frak{sl}_2$
\,with highest weight \,$\La$ \,and highest weight vector ${\bf v}_\La$\,:
\be
{\bf e\<\>v}_\La=0,\quad {\bf h\<\>v}_\La=\La\<\>{\bf v}_\La\,.
\ee
A basis of \;${\bf M}_\La$ \,is formed by the vectors \;${\bf f}^{\<\>r}{\bf v}_\La$,
\;$r\in\Z_{\geq 0}$\,. For \,$\La_1,\dots,\La_n\in\C$, a weight subspace
\bea
{\bf V}=\left(\bigotimes^n_{i=1} {\bf M}_{{\La_i}}\right)\left[\sum_{i=1}^n\La_i-2r\right]
\eea
has a basis formed by the vectors
\beq
\label{bas bf}
{\bf f}^{\<\>\vec r}:= {\bf f}^{\<\>r_1}{\bf v}_{\La_1} \ox\dots\ox {\bf f}^{\<\>r_n}{\bf v}_{\La_n}\,,
\eeq
where $\vec r\in\mc I_r$.

\subsection{Cyclotomic Hamiltonians on ${\bf V}$}

The classical trigonometric $r$-matrix in\\
$\End\<\>(\>{\bf M}_{\La_1}\ox {\bf M}_{\La_2})$ is defined by the formula
\be
\Omega_{\La_1,\La_2}(z)\,=\;\frac1{z-1}\>\Bigl(z\>{\bf e}\ox {\bf f}\>+\>{\bf f}\ox {\bf e}\>+\>
\frac{z+1}4\,({\bf h}\ox {\bf h}-\La_1\<\>\La_2)\Bigr).
\ee
Fix $\mu\in\C$. The trigonometric Gaudin Hamiltonians in $\End\bf V$ with parameter $\mu$
are defined by the formula
\be
H_m(z)\,=\,\mu\>({\bf h}^{(m)}\<-\La_m)\>+
\sum_{i\ne m}\,\Omega_{\La_m,\La_i}^{(m,i)}(z_m/z_i),
\qquad m=1,\dots,n.
\ee

Let \,$p=e^{2\pi ij/k}$, \,$1\leq j<k$\,, \,$(j,k)=1$\,, be a primitive root
of unity of order $k$. Let \,$\thi \in \C$ \,be such that \,$\thi^{\>k}=1$.

\vsk.2>

For $m=1,\dots,n$, define a linear operator $T_m$ acting on $\End{\bf V}$-valued functions of $z$ by the formula
\be
(T_m{\>*}X)(z)\,=\,\thi^{\>{\bf h}^{(m)}\!/2}\<X(p\<\>z_m)\,\thi^{-{\bf h}^{(m)}\!/2},
\qquad m=1,\dots, n.
\ee
Notice that \,$(\thi^{\>(\La_m-{\bf h}^{(m)}\<)\</2})^k\<=1$ \,and \;$T_m^k=1$\,.

Define the cyclotomic Gaudin Hamiltonians in $\End {\bf V}$ with parameters
$p,\,\thi,\,\mu$ by the formula
\bean
\label{Xmt2}
\tilde X_m(z)\,=\,\sum_{i=0}^{k-1}\,(T_m^{\>i}{\>*\,}H_m)(z)\,,
\qquad m=1,\dots,n.
\eean

\begin{thm}
[\cite{VY}]
\label{thm VY}
The cyclotomic Gaudin Hamiltonians commute. That is,
\bean
\label{VY c}
[\tilde X_l (z), \tilde X_m(z)] = 0,\qquad 1\leq l, m\leq n.
\eean
\end{thm}

In Section \ref{sec proof 1} we give a new proof of this statement, by deducing
the statement from formula \eqref{KLL=LL}.

\subsection{Diagonalization}

Denote
\begin{align}
\label{Xil}
\tilde\Xi_l(t;z)\,&{}=
\>-\>4\<\>\mu+2\<\>r-2-\sum_{i=1}^n\La_i\>+
\sum_{i=1}^n\,\frac{2\<\>\La_i}{1-t_l^k/z_i^k}\;-
\sum_{l'\<\ne\>l}\,\frac4{1-t_l^k/\<\>t_{l'}^k}\;,\qquad l=1,\dots,r\>,\kern-1em
\\[4pt]
\notag
\tilde\Zeta_m(t;z)\,&{}=\>\>
r\La_m-\sum_{i=1}^r\,\frac{2\<\>\La_m}{1-t_i^k/z_m^k}\;,\qquad
m=1,\dots,n\,.
\end{align}
Define the ${\bf V}$-valued weight function
\beq
\label{WFF}
{\bf W}(t;z)\,=\>\sum_{\vec r\in \mc I_r} {\bf W}_{\vec r}(t;z)\,{\bf f}^{\<\>\vec r}
\eeq
by the formula
\be
{\bf W}_{\vec r}(t;z)\,=\,
\prod_{i=1}^n\,\frac1{r_i\<\>!}\;
\on{Sym}_{\>t_1,\dots,t_r}\>\biggl(\,\prod_{i=1}^n\>
\prod_{l=\vec r^{(i-1)}+1}^{\vec r^{(i)}}
\frac{t_l/z_i}{1-t_l/z_i}\biggr)\,.
\ee

\begin{thm} [\cite{VY}]
\label{thm VY BA}

Let \,$p=e^{2\pi i j/k}$, \,$1\leq j<k$\,, \,$(j,k)=1$\,.
Let \;$\thi, \mu\in \C$ \,and \;$\thi^k=1$\,.
Let \;$(t^0;z^0)$ be a solution to the system of Bethe ansatz equations
\beq
\label{bae cy}
-\>4\<\>\mu+2\<\>r-2-\sum_{i=1}^n\La_i\>+
\sum_{i=1}^n\,\frac{2\<\>\La_i}{1-t_l^k/z_i^k}\;-
\sum_{l'\<\ne\>l}\,\frac4{1-t_l^k/t_{l'}^k}\;=\,0\,,
\eeq
$l=1,\dots,r$. Define a vector in \,${\bf V}$,
\beq
\label{Int cy}
J(t^0;z^0)\,=\,\sum_{c_1\<, \dots,\>c_r=0}^{k-1}
\thi^{\,\sum_{i=1}^r\<c_i}\;{\bf W}(p^{c_1}t_1^0, \dots, p^{c_r}t_r^0; z^0)\,.
\eeq
Then \,$J(t^0;z^0)$, if nonzero, \,is an eigenvector of the cyclotomic Gaudin Hamiltonians
\;$\tilde X_m(z^0)$\,, \,$m=1\lc n$\,, \>with respective eigenvalues
\beq
\label{ei cy}
\frac 12\>k\>r\La_m-\sum_{i=1}^r\,\frac{k\La_m}{1-(t^0_i/z^0_m)^k}\;.
\eeq
\end{thm}

\subsection{Proof of Theorem \ref{thm VY}}
\label{sec proof 1}

We identify the tensor products
of Verma modules over $U_q(\frak{sl}_2)$ and the corresponding tensor products
of Verma modules over $\frak{sl}_2$
by identifying the bases $(f^{\vec r})$ and
$(\<\>{\bf f}^{\vec r}\<\>)$.

\vsk.2>
In this way, we get a relation
between the trigonometric $R$-matrix $R_{\La_1,\La_2}(z;q)$
and the classical trigonometric $r$-matrix $\Omega_{\La_1,\La_2}(z)$
if \,$q=1-\eps$ \,and \,$\eps\to 0$\,,
\vvn.3>
\be
R_{\La_1,\La_2}(z;q)\,=\,1+2\>\eps\>\Omega_{\La_1,\La_2}(z)+\>\mc O(\eps^2)\,.
\vv.3>
\ee
Here the left-hand side is an $\End (M_{\La_1}\ox M_{\La_2})$-valued function, and the right-hand side
is an $\End\<\>(\>{\bf M}_{\La_1}\!\ox {\bf M}_{\La_2})$- valued function.

Fix two identified weight subspaces
\vvn.3>
\be
V\>=\,\left(\bigotimes^n_{i=1} M_{{\La_i}}\right)\left[\sum_{i=1}^n\La_i-2r\right],
\qquad
{\bf V}\>=\,\left(\bigotimes^n_{i=1} {\bf M}_{{\La_i}}\right)\left[\sum_{i=1}^n\La_i-2r\right].
\vv.3>
\ee
Consider the
$\End\bigotimes^n_{i=1} M_{{\La_i}}$-valued qKZ functions:
\vvn.2>
\be
K_m(z)\,:=\,K_m(z_1,\dots,z_n, p, q, q^{\La_1},\dots, q^{\La_n}, \ka),
\vv.3>
\ee
$m=1,\dots,n$, defined by formula \eqref{qkz op}. Suppose
\be
\ka\,=\,\thi\>(1-4\<\>\mu\>\eps)\,.
\ee
Then
\be
K_m(z)\,=\,\thi^{\>(\La_m-{\bf h}^{(m)}\<)\</2}\bigl(1+2\>\eps\>X_m(z)+\>o(\eps)\bigr)\,,
\ee
where
\beq
\label{Xm}
X_m(z)\,=\,\mu\>(\>{\bf h}^{(m)}\<\<-\La_m)\>+
\sum_{i=1}^{m-1}\,(T_m{\>*\,}\Omega_{\La_m\<,\>\La_i}^{(m,i)})(z_m/z_i)
+\!\<\sum_{l=m+1}^n\!\Omega_{\La_m\<,\>\La_l}^{(m,l)}(z_m/z_l)\,.
\vv.3>
\eeq
For the monodromy functions, this implies that
\vvn.4>
\be
\tilde K_m(z)\,=\,1+2\>\eps\>\tilde X_m(z)+\>o(\eps)\,.
\vv.3>
\ee
This formula and the commutativity of the monodromy functions $(\tilde K_m(z))$ imply
the commutativity of the cyclotomic Gaudin Hamiltonians $(\tilde X_m(z))$.
Theorem \ref{thm VY} is proved.

\subsection{Proof of Theorem \ref{thm VY BA}}
\label{sec proof 2}

Recall the functions \>\>$\Phi_l(t;z)$ \,and \>\>$\Psi_m(t;z)$
\,defined by \eqref{Phi}, \eqref{Psi}. Then
\begin{align}
\label{PhiX}
\Phi_l(t;z)\,&{}=\,\thi\,\prod_{i=1}^n\,\frac{1-t_l/z_i}{1-p\<\>t_l/z_i}\,
\bigl(\<\>1+\eps\>\Xi_l(t;z)+o(\eps)\bigr)\,,
\\[4pt]
\notag
\Psi_m(t;z)\,&{}=\,\prod_{i=1}^r\,\frac{1-t_i/z_m}{1-t_i/(pz_m)}\,
\bigl(\<\>1+\eps\>\Zeta_m(t;z)+o(\eps)\bigr)\,,
\end{align}
where
\begin{align}
\label{Xi}
\Xi_l(t;z)\,&{}=\>-\>4\<\>\mu\>+\>\sum_{i=1}^n\,\La_i\Bigl(
\frac{t_l/z_i}{1-t_l/z_i}\>+\>\frac1{1-p\<\>t_l/z_i}\Bigr)
\\[3pt]
\notag
&\hp{{}=\>-\>2\mu}\>-\>
2\>\sum_{l'\<\ne\>l}\,\Bigl(\frac1{1-t_l/t_{l'}}\>+\>\frac{p\<\>t_l/t_{l'}}{1-p\<\>t_l/t_{l'}}\Bigr)\,,
\\[4pt]
\notag
\Zeta_m(t;z)\,&{}=\>-\>\La_m\sum_{i=1}^r\,\Bigl(
\frac1{1-t_i/z_m}\>+\>\frac{t_i/z_m}{p-t_i/z_m}\Bigr)\,.
\end{align}
Due to the identities
\be
\frac k{1-x^k}\,=\,\sum_{c=0}^{k-1}\,\frac1{1-p^{\<\>c}\<\>x}\,,\qquad
\frac{k\<\>x^k}{1-x^k}\,=\,\sum_{c=0}^{k-1}\,\frac{p^{\<\>c}\<\>x}{1-p^{\<\>c}\<\>x}\;,
\vv-.6>
\ee
we have
\be
k\,\tilde\Xi_l(t;z)\,=\,\sum_{c=0}^{k-1}\,\Xi_l(p^{\<\>c}\<\>t_l;z)\,,\qquad
k\,\tilde\Zeta_m(t;z)\,=\,\sum_{c=0}^{k-1}\,\Zeta_m(t;p^{\<\>c}z_m)\,,
\vv.4>
\ee
where the functions \;$\tilde\Xi_l(t;z)$ \,and \;$\tilde\Zeta_m(t;z)$ \,are
defined by \eqref{Xil}. Thus for the corresponding monodromy functions, we have
\vvn.4>
\beq
\label{Xit}
\tilde\Phi_l(t;z)\,=\,1+\eps k\,\tilde\Xi_l(t;z)+o(\eps)\,,\qquad
\tilde\Psi_m(t;z)\,=\,1+\eps k\,\tilde\Zeta_m(t;z)+o(\eps)\,.
\vv.4>
\eeq
These asymptotics can be also obtained from formulae
\eqref{mon phi no}, \eqref{mon psi no}.

\smallskip
For the weight function defined by \eqref{4.8}, we have
\vvn.2>
\be
W_{\vec r}(t;z)\,=\,{\bf W}_{\vec r}(t;z)\,
\prod_{m=1}^n\,\prod_{l=1}^r\,(1-t_l/z_m)\>+\>o(1)\,.
\ee
For the functions \,$G_{l,m}(t;z)$ \,in integral representation \eqref{mIR},
Theorem \ref{thm 7.4} implies that
\vvn.3>
\be
G_{l,m}(t;z)\,=\,G_{l,m}^{\>\0}(t;z)\>+\>o(1)
\vv.3>
\ee
for suitable rational functions \,$G_{l,m}^{\>\0}(t;z)$\,.

\medskip

Theorem \ref{thm VY BA} follows from the described asymptotics as \,$\eps\to 0$
and Corollary \ref{cor eigen0}.

\begin{rem}

Under the assumptions of Theorem \ref{thm VY BA},
the function $J(t^0;z^0)$ has the additional properties:
\begin{align*}
%\label{last 1}
J(p\<\>t^0_l;z^0)\,&{}=\,J(t^0;z^0)\,, \qquad l=1,\dots,r\,,
\\
%\label{last 2*}
J(t^0;pz^0_m)\,&{}=\;\<\thi^{\>(\La_m-{\bf h}^{(m)}\<)\</2}\<\>J(t^0;z^0)\,,
\qquad m=1,\dots,n\,.
\end{align*}
These properties are corollaries of formulas \eqref{iliJ}, \eqref{ILIJ}.

\end{rem}

\section{Eigenvectors of monodromy operators form a basis}
\label{appB}

In Corollary~\ref{cor m eigen}, we considered the weight subspace
\[
V=\left(\bigotimes_{i=1}^n M_{\La_i}\right)\!\left[\sum_{i=1}^n\La_i-2r\right]
\]
and the qKZ discrete connection $(K_m)$ with multiplicative step
\begin{equation}
\label{pp}
p=e^{2\pi i j/k},\quad 1\leq j<k,\quad (j,k)=1.
\end{equation}
For any solution $(t^0;z^0)$ of the Bethe ansatz equations~\eqref{m bae},
we constructed a vector $I(t^0;z^0)\in V$ and showed that, whenever
this vector is nonzero, it is a common eigenvector of the monodromy
operators $(\widetilde K_m)$.

In this appendix, we restrict to the case $r=1$, so that
\[
V=\left(\bigotimes_{i=1}^n M_{\La_i}\right)
\!\left[\sum_{i=1}^n\La_i-2\right],
\qquad
\dim V=n,
\]
and $p$ is given by \eqref{pp}.
Under suitable genericity assumptions, we prove the following.

\begin{thm}
\label{thm:basis}
Fix $\La=(\La_1,\dots,\La_n)\in\C^n$.
For generic $z^0=(z_1^0,\dots,z_n^0)\in(\C^\times)^n$ and
$\ka\in\C^\times$, the Bethe ansatz equation has $n$
distinct solutions \,$u_1,\dots,u_n$. The corresponding vectors
$I(u_i;z^0)$, \,$i=1,\dots,n$, are nonzero and form a basis of \,$V$.
Moreover, each $I(u_i;z^0)$ is a common eigenvector of the monodromy
operators $(\widetilde K_m(z^0))$ with respective eigenvalues
\[
\frac{q^{k\La_m}-u_i^k/(z_m^0)^k}{1-q^{k\La_m}u_i^k/(z_m^0)^k}\,,
\qquad m=1,\dots,n.
\]
\end{thm}

The remainder of this appendix is devoted to the proof of
Theorem~\ref{thm:basis}.

\subsection{Rank-one functions}

For $r=1$, $t$ is a single variable and $z=(z_1,\dots,z_n)$.
The rank-one functions are
\begin{align}
\label{BPhi}
\Phi_1(t;z)&{}=\,\ka\prod_{i=1}^n\frac{1-q^{\La_i}t/z_i}
{q^{\La_i}-p\<\>t/z_i}\,,
\\[4pt]
\label{BPsi}
\Psi_m(t;z)&{}=\,\frac{q^{\La_m}-t/z_m}{1-q^{\La_m}t/(p z_m)}\,,
\qquad m=1,\dots,n.
\end{align}
The corresponding monodromy functions are
\begin{align}
\label{mo phi no}
\tilde\Phi_1(t;z)&{}=\,\ka^k\prod_{i=1}^n
\frac{1-q^{k\La_i}t^k/z_i^k}{q^{k\La_i}-t^k/z_i^k}\,,
\\[4pt]
\label{mo psi no}
\tilde\Psi_m(t;z)&{}=\,\frac{q^{k\La_m}-t^k/z_m^k}
{1-q^{k\La_m}t^k/z_m^k}\,,\qquad m=1,\dots,n.
\end{align}
The Bethe ansatz equation is \;$\tilde\Phi_1(t;z)=1$, that is,
\beq
\label{BAE}
\ka^k\prod_{i=1}^n\,\frac{1-q^{k\La_i}t^k/z_i^k}{q^{k\La_i}-t^k/z_i^k}\,=\,1\,.
\eeq

\subsection{Weight function}

Let
\[
f^{(m)}=v_{\La_1}\otimes\cdots\otimes
fv_{\La_m}\otimes\cdots\otimes v_{\La_n},
\qquad m=1,\dots,n.
\]
The weight function is
\beq
\label{BW1}
W(t;z)\,=\>\sum_{m=1}^n W_m(t;z)\,f^{(m)},
\eeq
where
\beq
\label{BW2}
W_m(t;z)\,=\,\frac{t}{z_m}
\prod_{l<m}\!\left(1-q^{\La_l}\frac{t}{z_l}\right)
\prod_{l>m}\!\left(q^{\La_l}-\frac{t}{z_l}\right).
\eeq
The Bethe vector~\eqref{mx0} takes the form
\beq
\label{mx1}
I(t;z)\,=\,\sum_{d=0}^{k-1}\,W(p^d t;z)\,
\prod_{a=0}^{d-1}\Phi_1(p^a t;z)\,.
\eeq
Note that \,$I(t;z)$ \,and $\widetilde K_m(z^0)$ \,also depend on $\La$
and $\ka$, but we suppress this dependence.

\smallskip
Corollary~\ref{cor m eigen} for $r=1$ takes the following form.

\begin{cor}
\label{cor:r=1}
Let $(t^0;z^0)$ be a solution to~\eqref{BAE}. Then $I(t^0;z^0)$,
whenever nonzero, is a common eigenvector of the monodromy operators
$\widetilde K_m(z^0)$, $m=1,\dots,n$, with eigenvalues
$\widetilde\Psi_m(t^0;z^0)$ given by~\eqref{mo psi no}.
\end{cor}

\subsection{Asymptotics as $\ka\to\infty$}

Fix $\La=(\La_1,\dots,\La_n)\in\C^n$ and
$z^0=(z_1^0,\dots,z_n^0)\in(\C^\times)^n$ such that
\begin{equation}
\label{B:ass}
\begin{aligned}
&q^{k\La_m}(z_m^0)^k \neq q^{-k\La_i}(z_i^0)^k,
&&1\leq i,m\leq n,\\
&q^{k\La_m}(z_i^0)^k \neq q^{k\La_i}(z_m^0)^k,
&&1\leq i\neq m\leq n.
\end{aligned}
\end{equation}
Set $\eps=\ka^{-1}$ and let $\eps\to 0$.

\smallskip

Rewrite~\eqref{BAE} as
\begin{equation}
\label{BAE1}
\prod_{i=1}^n\!
\left(1-q^{k\La_i}\frac{t^k}{(z_i^0)^k}\right)
=
\eps^k
\prod_{i=1}^n\!
\left(q^{k\La_i}-\frac{t^k}{(z_i^0)^k}\right).
\end{equation}
For $\eps=0$, this is a polynomial equation of degree $n$ in $t^k$
with $n$ distinct simple roots
\[
y_i^0:=q^{-k\La_i}(z_i^0)^k,\qquad i=1,\dots,n,
\]
which are distinct by the second condition in~\eqref{B:ass}, and at
which the right-hand side does not vanish by the first condition
in~\eqref{B:ass}. By the implicit function theorem, these roots deform
holomorphically for small $\eps$:
\[
y_i(\eps)
=
q^{-k\La_i}(z_i^0)^k+\mc O(\eps^k).
\]
For each $i$, choose a $k$-th root $u_i(\eps)$ satisfying
\begin{equation}
\label{ass:t}
u_i(\eps)=q^{-\La_i}z_i^0+\mc O(\eps^k).
\end{equation}
These give $n$ solutions of~\eqref{BAE} for small $\eps$.

\begin{lem}
\label{lem:B1}
For $i=1,\dots,n$ and \,$d=1,\dots,k-1$,
\[
\prod_{c=0}^{d-1}\,\Phi_1(p^{\<\>c} u_i(\eps);z^0)
=\mc O(\eps^{k-d})
\quad\text{as}\;\;\eps\to0\,.
\]
In particular, this product tends to zero.
\end{lem}

\begin{proof}
By~\eqref{BPhi} and~\eqref{ass:t}, the factor with $c=0$ satisfies
\[
\Phi_1(u_i(\eps);z^0)=\mc O(\eps^{k-1}),
\]
because the numerator factor $1-q^{\La_i}u_i(\eps)/z_i^0\to0$
while $\ka=\eps^{-1}\to\infty$.
For $c=1,\dots,k-1$, the argument $p^{\<\>c} u_i(\eps)\to p^{\<\>c}
q^{-\La_i}z_i^0$, which is not a zero of the numerator since
\[
1-q^{\La_i}p^{\<\>c} q^{-\La_i}z_i^0/z_i^0=1-p^{\<\>c}\ne0
\]
for $c\ne0$.
Hence $\Phi_1(p^{\<\>c} u_i(\eps);z^0)=\mc O(\eps^{-1})$ for each such $c$.
Multiplying all factors gives
\[
\prod_{c=0}^{d-1}\Phi_1(p^{\<\>c} u_i(\eps);z^0)
=\mc O(\eps^{k-1})\cdot\mc O(\eps^{-1})^{d-1}=\mc O(\eps^{k-d}).
\]
\vvn-2.2>\\\strut
\end{proof}

\begin{lem}
\label{lem:B2}
For $1\leq i,m\leq n$ and \,$d=0,\dots,k-1$,
\[
\lim_{\eps\to 0} W_m(p^d u_i(\eps);z^0)\,=\,W_m(p^d q^{-\La_i}z_i^0;z^0).
\]
Moreover,
\begin{align}
W_m(q^{-\La_i}z_i^0;z^0)\,&{}=\,0\,, \qquad m>i,
\label{B15}
\\
W_i(q^{-\La_i}z_i^0;z^0)\,&{}=\,q^{-\La_i}
\prod_{l<i}\left(1-q^{\La_l-\La_i}\frac{z_i^0}{z_l^0}\right)\>
\prod_{l>i}\left(q^{\La_l}-q^{-\La_i}\frac{z_i^0}{z_l^0}\right)\>\ne\,0\,.
\label{B16}
\end{align}
\end{lem}

\begin{proof}
The limit follows directly from~\eqref{BW2} and~\eqref{ass:t}.\\
For $m>i$, the product $\prod_{l<m}$ in $W_m$ contains the factor
$l=i$, which gives $1-q^{\La_i}q^{-\La_i}z_i^0/z_i^0=0$.
Hence $W_m(q^{-\La_i}z_i^0;z^0)=0$.\\
For $m=i$, substituting $t=q^{-\La_i}z_i^0$ into~\eqref{BW2} gives
formula~\eqref{B16}. The factors are nonzero by~\eqref{B:ass}.
\end{proof}

\begin{cor}
\label{cor:lim}
For $i=1,\dots,n$,
\begin{equation}
\label{lim:I}
\lim_{\eps\to 0} I(u_i(\eps);z^0)
=
W(q^{-\La_i}z_i^0;z^0).
\end{equation}
The vectors $W(q^{-\La_i}z_i^0;z^0)$, $i=1,\dots,n$, form a basis of \;$V$.
\end{cor}

\begin{proof}
From~\eqref{mx1},
\[
I(u_i(\eps);z^0)\,=\,W(u_i(\eps);z^0)\>+\>
\sum_{d=1}^{k-1}\,W(p^d u_i(\eps);z^0)\,
\prod_{c=0}^{d-1}\,\Phi_1(p^{\<\>c} u_i(\eps);z^0)\,.
\]
The $d=0$ term tends to $W(q^{-\La_i}z_i^0;z^0)$ by
Lemma~\ref{lem:B2}. For $d\geq1$, the factor
\[
\prod_{c=0}^{d-1}\,\Phi_1(p^{\<\>c} u_i(\eps);z^0)\,=\,
\mc O(\eps^{k-d})\>\to\>0
\]
by Lemma~\ref{lem:B1}, while $W(p^d u_i(\eps);z^0)$ remains bounded.
Hence every term with $d\geq1$ tends to zero, giving~\eqref{lim:I}.

For linear independence, the matrix
$(W_m(q^{-\La_i}z_i^0;z^0))_{1\leq\<\>i,\<\>m\leq n}$
is lower triangular by~\eqref{B15}, with nonzero diagonal entries
by~\eqref{B16}. Hence it is invertible.
\end{proof}
\subsection{Proof of Theorem~\ref{thm:basis}}

Write $y=t^k$. The Bethe equation~\eqref{BAE} becomes, after clearing
denominators,
\begin{equation}
\label{poly:bethe}
\ka^k
\prod_{i=1}^n\!
\left(1-q^{k\La_i}\frac{y}{(z_i^0)^k}\right)
=
\prod_{i=1}^n\!
\left(q^{k\La_i}-\frac{y}{(z_i^0)^k}\right),
\end{equation}
a polynomial equation of degree $n$ in $y$. Its discriminant
$\Delta(\ka)$ with respect to $y$ is a nonzero polynomial in $\ka$:
it is nonzero for large $|\ka|$ by conditions~\eqref{B:ass}, so
$\Delta\not\equiv0$. Hence $\Delta(\ka)=0$ for at most finitely many
values of $\ka$.

For $\ka$ outside this finite exceptional set, equation~\eqref{poly:bethe}
has $n$ distinct roots \,$y_1(\ka),\dots,\allowbreak y_n(\ka)$, which are branches
of algebraic functions of $\ka$. For each $i$, fix a continuous branch
of the $k$-th root $u_i(\ka)=y_i(\ka)^{1/k}$, chosen so that
$u_i(\ka)\to q^{-\La_i}z_i^0$ as $\ka\to\infty$, in accordance
with~\eqref{ass:t}. The vector $I(u_i(\ka);z^0)$ is then a locally
holomorphic and globally algebraic function of $\ka$ along this
branch.

The determinant
\[
D(\ka)
:=
\det\bigl[I(u_1(\ka);z^0)\ \cdots\ I(u_n(\ka);z^0)\bigr]
\]
is therefore an algebraic function of $\ka$, defined on a Riemann
surface $\mc R$ over $\C^\times$. By Corollary~\ref{cor:lim},
\[
D(\ka)\longrightarrow D_0\ne0
\quad\text{as }\ka\to\infty
\]
along the branch constructed above. In particular, $D$ is not
identically zero on $\mc R$. Since $\mc R$ is a connected Riemann
surface and $D$ is analytic and not identically zero, its zeros form a
discrete subset of $\mc R$. Therefore $D(\ka)\ne0$ for all $\ka$
outside a discrete (in fact finite) exceptional set, that is, for
generic $\ka\in\C^\times$.

For such generic $\ka$, the $n$ vectors $I(u_i;z^0)$ are linearly
independent and hence form a basis of $V$. In particular, each
$I(u_i;z^0)$ is nonzero. The eigenvalue statement follows from
Corollary~\ref{cor:r=1}.\qed

\begin{rem}
Generically, for any $m$ and $i\ne l$, the eigenvalues of $\tilde K_m(z^0)$ on
$I(u_i;z^0)$ and $I(u_l;z^0)$ are distinct, see the formula for eigenvalues in
Theorem \ref{thm:basis}.
\end{rem}

\end{document}